\documentclass[submission,copyright,creativecommons]{eptcs}
\providecommand{\event}{AiML 2026} 

\usepackage{aiml26}

\usepackage{iftex}

\ifpdf
  \usepackage{underscore}         % Only needed if you use pdflatex.
  \usepackage[T1]{fontenc}        % Recommended with pdflatex
\else
  \usepackage{breakurl}           % Not needed if you use pdflatex only.
\fi

\title{Uniform Lyndon Interpolation via Non-wellfounded Proofs\thanks{Research supported by the Swiss National Science Foundation project 200021\_214820.}}
\author{Borja Sierra Miranda
\institute{University of Bern\\ Bern, Switzerland}
\email{borja.sierra@unibe.ch}
\and
Thomas Studer
\institute{University of Bern\\ Bern, Switzerland}
\email{thomas.studer@unibe.ch}
}

\newcommand{\titlerunning}{Uniform Lyndon Interpolation via Non-wellfounded Proofs}
\newcommand{\authorrunning}{B. Sierra Miranda and T. Studer}

\hypersetup{
  bookmarksnumbered,
  pdftitle    = {\titlerunning},
  pdfauthor   = {\authorrunning},
  pdfsubject  = {EPTCS},               % Consider adding a more appropriate subject or description
  pdfkeywords = {keyword1, keyword2} % Uncomment and enter keywords specific to your paper
}

\usepackage{xcolor}
\usepackage{bussproofs}
\usepackage{modalops}
\usepackage{amssymb}
\usepackage{multicol}

\newcommand\set[1]{\left\{#1\right\}}
\newcommand\union{\cup}
\newcommand\inter{\cap}
\newcommand\Union{\bigcup}

\newcommand\KT{\mathsf{K4}}
\newcommand\GL{\mathsf{GL}}
\newcommand\GLS{\mathsf{GLS}}
\newcommand\PA{\mathsf{PA}}
\newcommand\IL{\mathsf{IL}}

\newcommand\lgh{\mathrm{lgh}} % Length
\newcommand\lhg{\mathrm{lhg}} % Local height
\newcommand\satc[1]{|#1|_{\mathrm{Sat}}} % Saturation complexity
\newcommand\hg{\mathrm{hg}}
\newcommand\lrul{\mathrm{lRul}}
\newcommand\sub{\mathrm{Sub}}
\newcommand\voc{\mathrm{Voc}}

\newcommand\g[1]{\mathcal{G}#1}
\newcommand\n[1]{\mathcal{G}^\infty #1}

\newcommand\MP{\mathrm{MP}}
\newcommand\NEC{\mathrm{NEC}}
\newcommand\K{\mathrm{K}}
\renewcommand\L{\mathrm{L}}
\renewcommand\S{\mathrm{T}}

\newcommand\ax{\mathrm{ax}}
\newcommand\Ax{\mathrm{Ax}}
\newcommand\botL{\bot\mathrm{L}}
\newcommand\botR{\bot\mathrm{R}}
\newcommand\toL{{\to}\mathrm{L}}
\newcommand\toR{{\to}\mathrm{R}}

\newcommand\wedgeL{{\wedge}\mathrm{L}}
\newcommand\wedgeR{{\wedge}\mathrm{R}}
\newcommand\negL{{\neg}\mathrm{L}}
\newcommand\negR{{\neg}\mathrm{R}}
\newcommand\toLsat{{\to}\mathrm{L}_{\mathrm{sat}}}
\newcommand\toRsat{{\to}\mathrm{R}_{\mathrm{sat}}}
\newcommand\Ssat{\S_{\mathrm{sat}}}
\newcommand\modal[1]{\nec^{#1}}
\newcommand\interpolation[1]{\nec^{#1}_{\mathrm{sat}}}
\newcommand\lob{\mathrm{L\ddot{o}b}}

\newcommand\rep{\mathrm{Rep}}
\newcommand\wk{\mathrm{Wk}}
\newcommand\inv{\mathrm{inv}}
\newcommand\ctr{\mathrm{Ctr}}
\newcommand\cut{\mathrm{Cut}}
\newcommand\emp{\mathrm{Emp}}

\newcommand\domain{\mathrm{Dom}}

\begin{document}
\maketitle

\begin{abstract}
Non-wellfounded proof theory has been applied to establish uniform interpolation and Lyndon interpolation (separately) for multiple logics.
However, it has not yet been used to prove uniform Lyndon interpolation.
We close this gap by showing uniform Lyndon interpolation for the provability logic \(\GLS\). This logic was known to have uniform interpolation, but it was open whether it has uniform Lyndon interpolation (or at least non-uniform Lyndon interpolation).
The methodology we provide is easy to adapt to other provability logics if a non-wellfounded sequent calculus is available for them.
In addition, we offer an alternative proof of cut elimination for \(\GLS\) via non-wellfounded proofs.
\end{abstract}

\section{Introduction}

Non-wellfounded proof theory provides a powerful machinery to study of logics with explicit fixpoints, see e.g.~\cite{Brotherston, thomasCK, KokkinisStuder+2016+171+192, saurin, guillermo-CTL, ill-founded-intui-linear-time, das, master-modality}.
However, its use is not restricted to these logics.
In particular,  provability logics are closely connected to a simple class of non-wellfounded proof systems~\cite{shamkanovGl, shamkanovGrz, justus, coalgebraic, proofth-interpretability, bimodal-provability-proofth}.

Non-wellfounded sequent calculi have  been successfully applied to obtain many  interpolation results.
In particular, they have been used to establish Craig interpolation \cite{pdl-interpolation, converse-pdl-interpolation}, Lyndon interpolation \cite{shamkanovGl} and uniform interpolation \cite{interpolation-guillermo, uip-interpretability}.
Still, the combination of uniform and Lyndon interpolation 
(the so-called uniform Lyndon interpolation property~\cite{Akbar,Kurahashi})
has not been studied with these methods.
This paper is the first investigation in that direction, using non-wellfounded proofs to establish uniform Lyndon interpolation of the logic \(\GLS\).

Solovay's logic \(\GLS\) (also known simply as \(\mathsf{S}\)) is the provability logic corresponding to the provability predicate of \(\PA\) according to true arithmetic \cite{solovay}, and it is known to have Craig interpolation~\cite{boolos-GLS, lev-GLS, GLSprooftheory}.\footnote{Uniform interpolation for \(\GLS\) follows easily  from the fact that \(\GL\) has uniform interpolation, by the same proof as in \cite{lev-GLS}. Note though that the translation between \(\GL\) and \(\GLS\) described in that paper does not yield Lyndon interpolation, as it does not preserve the polarity of variables.}
No proof that it has Lyndon or uniform Lyndon interpolation is known.
Usually for such a proof the use of non-wellfounded proof theory is beneficial, as non-wellfounded sequent calculi get rid of what is usually called the \emph{diagonal formula},  which, in wellfounded proofs, does not preserve the polarity.
In this paper we will show that, indeed, \(\GLS\) has uniform Lyndon interpolation using non-wellfounded proofs.

\paragraph{Contributions.}
In this paper the contributions are three-fold.
\begin{enumerate}
  \item An alternative proof of cut elimination for \(\GLS\) is provided using non-wellfounded proofs (the original proof using ideas from the wellfounded cut elimination of \(\GL\) can be found in \cite{lev-GLS} and more detailed in \cite{GLSprooftheory}).
  \item The introduction of the concepts of Lyndon fixpoints and Lyndon equational systems, which are necessary to show that the uniform interpolant created via non-wellfounded proofs respects polarities.
  \item The first proof that \(\GLS\) has uniform Lyndon interpolation.
\end{enumerate}

\paragraph{Organization.}
In Section~\ref{sec:preliminaries} we will introduce the concepts needed from non-wellfounded proof theory.
In Section~\ref{sec:sequent-calculi} we will give the definition of \(\GLS\) and its sequent calculi, with an alternative proof of cut elimination using non-wellfounded proofs.
In Section~\ref{sec:ulip} we will introduce the notions of Lyndon fixpoint and Lyndon equational system, which are at the heart of our proof of uniform Lyndon interpolation.
We conclude by proving the promised result, uniform Lyndon interpolation for \(\GLS\).

\section{Preliminaries: local progress proof theory}\label{sec:preliminaries}

Let us start by fixing the notions concerning trees.
Given a set \(X\) we will write \(X^*\) (\(X^+\)) to mean the set of (non-empty) finite sequences of elements in \(X\), \(\epsilon\) will denote the empty sequence.
As usual, we will write \(w \leq v\) to mean that \(w\) is an initial prefix of \(v\) and \((w,v]\) to mean the set \(\set{u \in X^* \mid w < u \leq v}\).
A (finitely branching) tree on \(A\) is a function \(T\) whose image is contained in \(A\) and whose domain is a prefix-closed non-empty subset of \(\mathbb{N}^*\) such that for every \(w \in \domain(T)\) there is an unique natural number \(k\) (called the \emph{arity of \(w\) in \(T\)}) such that \(wi \in \domain(T)\) iff \(i < k\).
The \(wi \in \domain(T)\) are also called the \emph{immediate successors of \(w\)}.
Note that the domain of any tree always contains the word \(\epsilon\), which is called the \emph{root of \(T\)}.

The elements of \(\domain(T)\) are also called the \emph{nodes of \(T\)}, the \(0\)-ary nodes are called \emph{leaves} and the rest of the nodes are called \emph{interior nodes}.
Finally, an \emph{infinite branch in a tree \(T\)} is sequence of nodes \((w_i)_{i \in \mathbb{N}}\) such that \(w_0 = \epsilon\) and for each \(i\) there is a \(j\) such that \(w_{i+1} = w_i j\).

Fix a set \(\text{Seq}\) whose elements we will call \emph{sequents}, a sequent rule is a subset of \(\text{Seq}^+\).
Given a rule \(R\), we say that \((S_0,\ldots,S_{n-1},S)\) is an \emph{instance of \(R\) with premises \(S_0,\ldots,S_{n-1}\) and conclusion \(S\)} if \((S_0,\ldots,S_{n-1},S) \in R\).
A rule \(R\) is said to be \emph{\(n\)-ary} if \(R \subseteq \text{Seq}^{n+1}\).
We introduce the kind of sequent calculi we are going to use, called \emph{local progress sequent calculi}.

\begin{definition}
  A \emph{local progress sequent calculus} is a pair \(\mathcal{G} = (\mathcal{R}, (L_R)_{R \in \mathcal{R}})\) where \(\mathcal{R}\) is a set of sequent rules and each \(L_R\) is a function that takes an instance \(r = (S_0,\ldots,S_{n-1},S)\) of \(R\) and returns a subset of \(\set{0,\ldots,n-1}\).
  \(\mathcal{G}\) is said to be \emph{wellfounded} if each \(L_R\) is the constant function returning \(\varnothing\).
\end{definition}

We are prepared to define the notion of proof in a local progress calculus.
From the defintion of proof we can infer that a wellfounded sequent calculus is just a sequent calculus in the usual (wellfounded) proof theory.

\begin{definition}
  Given a local progress sequent calculus \(\mathcal{G} = (\mathcal{R}, (L_R)_{R \in \mathcal{R}})\), a \emph{preproof in \(\mathcal{G}\)} is a tree with labels in \(\text{Seq} \times \mathcal{R}\) such that for each node \(w\) with immediate successors \(w0,\ldots,w(n-1)\) we have that \((S_0,\ldots,S_{n-1},S) \in R\) where \(S\) is the sequent at \(w\), \(R\) is the rule at \(w\) and \(S_i\) is the sequent at \(wi\).

  Let \(w\) be a node in \(\pi\) with immediate successors \(w0,\ldots,w(n-1)\).
  We say that \(wi\) is a \emph{progressing node} if \(i \in L_R(r)\) where \(R\) is the rule at \(w\) and \(r = (S_0,\ldots,S_{n-1},S)\) where \(S\) is the sequent at \(w\) and \(S_i\) the sequent at \(wi\).
  A proof is a preproof in which any infinite branch has infinitely many progressing nodes.
\end{definition}

We will write \(\mathcal{G} \vdash S\) to mean that \(S\) is provable in \(\mathcal{G}\) and \(\pi \vdash_{\mathcal{G}} S\) to mean that \(\pi\) is a proof of \(S\) in \(\mathcal{G}\), omitting the subscript \(_\mathcal{G}\) when it is clear from context.
It will be common to write an instance \((S_0,\ldots,S_{n-1},S)\) of a rule \(R\) as
\[
  \AxiomC{\(S_0\)}
  \AxiomC{\(\cdots\)}
  \AxiomC{\(S_{n-1}\)}
  \RightLabel{\(R\)}
  \TrinaryInfC{\(S\)}
  \DisplayProof
\]
Given a proof \(\pi\) in \(\mathcal{G}\) we will define its main local fragment as the finite tree obtained from cutting the tree at the first progressing nodes from the root (in particular removing the progressing nodes).
The \emph{local height of \(\pi\)}, denoted \(\lhg(\pi)\), is the height of its main local fragment.
The \emph{local rules of \(\pi\)}, denoted \(\lrul(\pi)\), is the set of rules occuring in the main local fragment.
We will say that \(\pi\) is \emph{locally \(R\)-free} if \(R \not \in\lrul(\pi)\).
We note that if \(\mathcal{G}\) is a wellfounded calculus the notion of local height agrees with the usual notion of height and the local rules is the set of rules occuring in the proof.

Given a local progress calculus \(\mathcal{G} = (\mathcal{R},(L_R)_{R \in \mathcal{R}})\) and a rule \(R'\) we define \(\mathcal{G} + R'\) as the local progress calculus with rules \(\mathcal{R} \union \set{R'}\) and \(L_{R'}\) the constant function returning \(\varnothing\).
Rules can interact with local progress calculi in different ways, the following definition introduce some of these ways and the theorem below shows that some of them are equivalent.

\begin{definition}\label{def:prop-of-rules}
  Let \(\mathcal{G}\) be a sequent calculus and \(R\) be a rule.
  We say that: %\(R\) is
  \begin{enumerate}
    \item \(R\) is admissible if for any instance \((S_0,\ldots,S_{n-1},S) \in R\) we have that \(\pi_i \vdash_{\mathcal{G}} S_i\) for \(i < n\) implies the existence of \(\pi \vdash_{\mathcal{G}} S\).
      In addition we say that \(R\) is \emph{admissible preserving local height} if \(\lhg(\pi) \leq \max_{i < n}\lhg(\pi_i)\) and
        \emph{admissible preserving local rules} if \(\lrul(\pi) \subseteq \Union_{i < n}\lrul(\pi_i)\).
    \item \(R\) is eliminable if \(\mathcal{G} + R\vdash S\) implies \(\mathcal{G} \vdash S\).\footnote{
In the proof-theoretic tradition, it is common to study the concept of effective cut elimination, i.e.~to study who to get rid of the cut rule in proofs using computable (effective) means.
Often, cut-elimination is understood as effective cut-elimination and non-effective cut-elimination is called cut admissibility.
This is justified since eliminablity and admissibility,  as given in Definition~\ref{def:prop-of-rules}, are equivalent for wellfounded proofs.  However,  this equivalence does not hold for non-wellfounded proofs.  The best approximation we can get is Theorem~\ref{th:local-adm-and-eliminability}. 
For this reason, we will not follow the tradition.  For us, a rule being eliminable means exactly what it says: that the rule can be eliminated from the calculus without affecting the provability of sequents.
In case we want to talk about the effectiveness of our methodology, we will say effective cut elimination explicitely.
}
    \item \(R\) is locally admissible if for any instance \((S_0,\ldots,S_{n-1},S) \in R\) we have that if there are locally \(R\)-free \(\pi_i \vdash_{\mathcal{G}} S_i\)  for \(i < n\) then there is a locally \(R\)-free \(\pi \vdash_{\mathcal{G}} S\).
    \item An \(n\)-ary rule \(R\) is \(i\)-invertible  in \(\mathcal{G}\) (where \(i < n\)) if the rule
      \(
        \set{(S,S_i) \mid (S_0,\ldots,S_{n-1},S) \in R}
      \)
      is admissible.
      \(R\) is \emph{invertible} if it is \(i\)-invertible for each \(i < n\).
  \end{enumerate}
We will talk about invertibility preserving local height and/or local rules with the obvious meaning.
\end{definition}
Given an admissible rule \(R\) in \(\mathcal{G}\) an instance \((S_0,\ldots,S_{n-1},S) \in R\) and proofs \(\pi_i \vdash_{\mathcal{G}} S_i\) for \(i < n\) we will write \(R(\pi_0,\ldots,\pi_{n-1})\) to mean a proof of \(S\) in \(\mathcal{G}\) that exists by admissibility.
For the invertibility of a rule \(R\) we will write instead \(\inv{R}(\pi)\).

The following result allows an easy development of local progress proof theory (for the details see~\cite{proofth-interpretability,coalgebraic}).

\begin{theorem}\label{th:local-adm-and-eliminability}
  Let \(\mathcal{G}\) be a local progress sequent calculus, then \(R\) is eliminable in \(\mathcal{G}\) iff \(R\) is locally admissible in \(\mathcal{G}\).
  If \(\mathcal{G}\) is wellfounded, then both are equivalent to \(R\) is admissible in \(\mathcal{G}\).
\end{theorem}

Finally, sometimes we will need to work with non-wellfounded proofs that have a particular finite representation.
To define this notion we introduce the notion of trees with backedges.
A \emph{tree with backedges} is an ordered pair \(\tau = (T,(\cdot)^\circ)\) such that \(T\) is a finite tree and \((\cdot)^\circ\) is a function from a subset of the leafs of \(\tau\) to the nodes of \(\tau\) such that if \(w\) is in the domain then \(w^\circ < w\), the nodes in the domain of \((\cdot)^\circ\) are called \emph{repeat nodes}.

\begin{definition}
  Given a local progress calculus \(\mathcal{G} = (\mathcal{R}, (L_R)_{R \in \mathcal{R}})\), a \emph{cyclic preproof in \(\mathcal{G}\)} is a tree with backedges \(\pi = (T, (\cdot)^\circ)\) on \(\text{Seq} \times \mathcal{R}\) such that for any non-repeat node \(w\) with immediate successors \(w0, \ldots, w(n-1)\), we have that \((S_0,\ldots,S_{n-1},S) \in R\) where \(S\) is the sequent at \(w\), \(R\) is the rule at \(w\) and \(S_i\) is the sequent at \(wi\).
    and for any repeat node \(w\), \(w\) and \(w^\circ\) have the same sequent and rule.
    A \emph{cyclic proof in \(\mathcal{G}\)} is a cyclic preproof such that for any repeat node \(w\) there is a progressing node in \((w^\circ,w]\).
\end{definition}
In cyclic preproofs we will annotate repeat nodes with the word \(\rep\) at the right, instead of the rule at the node.

\section{The logic GLS and its sequent calculi}\label{sec:sequent-calculi}

We fix an infinite countable set \(\text{Var}\) whose elements are called \emph{propositional variables}.

\begin{definition}
  We define the language \(\mathcal{L}_{\nec}\) given by the following Backus-Naur form:
  \[
    \phi ::= p \mid \bot \mid \phi \to \phi \mid \nec \phi,
  \]
  where \(p \in \text{Var}\).
  The expressions of \(\mathcal{L}_{\nec}\) are called \emph{(modal) formulas}. 
  The complexity of a formula \(\phi\) is the number of logical connectives of \(\phi\), and will be denoted as \(|\phi|\).
\end{definition}

Given a multisets of formulas \(\Gamma\), we will  write \(\nec \Gamma\) to mean the multiset \(\set{\nec \phi \mid \phi \in \Gamma}\) and \(\necd \Gamma\) to mean the multiset \(\Gamma \union \nec \Gamma\).
  Then, for a formula \(\phi\), \(\necd \phi\) will be the multiset with two elements \(\set{\phi, \nec \phi}\).

\begin{definition}
  We define the logic \(\GL\) as the smallest set of formulas such that
  \begin{enumerate}
    \item every classical propositional tautology in \(\mathcal{L}_{\nec}\) is in \(\GL\),
    \item \((\K)\) \(\nec(\phi \to \psi) \to \nec \phi \to \nec \psi\) is in \(\GL\),
    \item \((\L)\) \(\nec(\nec \phi \to \phi) \to \nec \phi\) is in \(\GL\),
  \end{enumerate}
  and that is closed under the rules
  \[
    \AxiomC{\(\phi\)}
    \AxiomC{\(\phi \to \psi\)}
    \RightLabel{\(\MP\)}
    \BinaryInfC{\(\psi\)}
    \DisplayProof
    \qquad
    \AxiomC{\(\phi\)}
    \RightLabel{\(\NEC\)}
    \UnaryInfC{\(\nec \phi\)}
    \DisplayProof
  \]
  The elements of \(\GL\) are also called the \emph{theorems of\/ \(\GL\)}.
  We will write \(\GL \vdash \phi\) to mean \(\phi \in \GL\).

  We define the logic \(\GLS\) as the smallest set of formulas such that
  \begin{enumerate}
    \item every theorem of \(\GL\) is in \(\GLS\),
    \item \((\S)\) \(\nec \phi \to \phi\) is in   \(\GLS\),
  \end{enumerate}
  and that is closed under the rule \((\MP)\).
  The elements of \(\GLS\) are also called the \emph{theorems of\/ \(\GLS\)}.
  We will write \(\GLS \vdash \phi\) to mean \(\phi \in \GLS\).
\end{definition}

Now we will define sequent calculi for \(\GLS\). We start by fixing the notion of sequent we will work with.
A \emph{sequent} is a triple \((\Gamma, \Delta,i)\) where \(\Gamma, \Delta\) are multisets of formulas and \(i \in \set{0,1}\).
We  write \((\Gamma, \Delta, 0)\) as \(\Gamma \Rightarrow \Delta\) and \((\Gamma, \Delta, 1)\) as \(\Gamma \Rrightarrow \Delta\).
We  use \(\gg\), possibly with subindexes, to mean either \( \Rightarrow \) or \(\Rrightarrow\).
Given a multiset \(\Sigma\),  we  let \(\Sigma^s\) be the \emph{set} with the same elements as \(\Sigma\) (but without repetitions).

\begin{figure}
  \[
    \AxiomC{}
    \RightLabel{\(\ax\)}
    \UnaryInfC{\(p, \Gamma \gg p, \Delta\)}
    \DisplayProof
    \qquad
    \AxiomC{}
    \RightLabel{\(\botL\)}
    \UnaryInfC{\(\bot, \Gamma \gg \Delta\)}
    \DisplayProof
    \qquad
    \AxiomC{\(\Gamma \gg \Delta\)}
    \RightLabel{\(\botR\)}
    \UnaryInfC{\(\Gamma \gg \Delta, \bot\)}
    \DisplayProof
  \]

  \[
    \AxiomC{\(\Gamma \gg \phi, \Delta\)}
    \AxiomC{\(\psi, \Gamma \gg \Delta\)}
    \RightLabel{\(\toL\)}
    \BinaryInfC{\(\phi \to \psi, \Gamma \gg \Delta\)}
    \DisplayProof
    \qquad
    \AxiomC{\(\phi, \Gamma \gg \psi, \Delta\)}
    \RightLabel{\(\toR\)}
    \UnaryInfC{\(\Gamma \gg \phi \to \psi, \Delta\)}
    \DisplayProof
  \]

  \[
    \AxiomC{\(\necd \Sigma, \nec \phi \Rightarrow \phi\)}
    \RightLabel{\(\modal{\GL}\)}
    \UnaryInfC{\(\nec \Sigma, \Gamma \gg \nec\phi, \Delta\)}
    \DisplayProof
    \qquad
    \AxiomC{\(\necd \Sigma \Rightarrow \phi\)}
    \RightLabel{\(\modal{\KT}\)}
    \UnaryInfC{\(\nec \Sigma, \Gamma \gg \nec\phi, \Delta\)}
    \DisplayProof
    \qquad
    \AxiomC{\(\necd \phi, \Gamma \Rrightarrow \Delta\)}
    \RightLabel{\(\S\)}
    \UnaryInfC{\(\nec \phi, \Gamma \Rrightarrow \Delta\)}
    \DisplayProof
  \]

  \caption{Sequent Rules}
  \label{fig:sequent-rules}
\end{figure}

  We introduce some terminology related to the rules of Figure~\ref{fig:sequent-rules}.
  In \((\botR)\), \((\toL)\), \((\toR)\), \((\S)\), \((\modal{\GL})\) and \((\modal{\KT})\) the displayed formula at the conclusion is called \emph{principal formula}.
  In \((\modal{\GL})\) and \((\modal{\KT})\) the formulas at the conclusion in \(\nec \Sigma\) will be called \emph{auxiliary}.
  In \((\modal{\GL})\) the formula \(\nec \phi\) at the premise is called \emph{diagonal formula}.
  In \((\ax)\), \((\botL)\), \((\modal{\GL})\) and \((\modal{\KT})\) the formulas in \(\Gamma\) and \(\Delta\) will be said to belong to the \emph{weakening part}. Note that the weakneing part of any rule instance can be changed arbitrarily and still we will have a rule instance of the same rule.

\begin{definition}
  We make the following definitions.
  \begin{enumerate}
    \item \(\g{\GLS}\) is the wellfounded calculus with rules \((\ax)\), \((\botL)\), \((\botR)\), \((\toL)\), \((\toR)\), \((\modal{\GL})\) and \((\S)\).\footnote{This calculus is an minor variation of the ones appearing in \cite{lev-GLS, GLSprooftheory}.}
    \item \(\n{\GLS}\) is the local progress calculus with rules \((\ax)\), \((\botL)\), \((\botR)\), \((\toL)\), \((\toR)\), \((\modal{\KT})\) and \((\S)\).
      Progress only occurs at the premise of \((\modal{\KT})\).
  \end{enumerate}
\end{definition}

\begin{lemma}\label{lm:basic-prop-sequent}
  For any formula \(\phi\) and multisets \(\Gamma, \Delta\) we have that
  \begin{enumerate}
    \item \(\g{\GLS} \vdash \phi, \Gamma \gg \phi, \Delta\) and \(\n{\GLS} \vdash \phi, \Gamma \gg \phi, \Delta\).
      We denote the use of this inside proofs as \((\Ax)\).
    \item The following rule is admissible in \(\g{\GLS}\; (+\cut)\) and in \(\n{\GLS}\; (+\cut)\):
      \[
        \AxiomC{\(\Gamma \Rightarrow \Delta\)}
        \RightLabel{\(\mathrm{Prom}\)}
        \UnaryInfC{\(\Gamma \Rrightarrow \Delta\)}
        \DisplayProof
      \]
  \end{enumerate}
\end{lemma}
\begin{proof}
  The proof of 1.\ is by induction on \(\phi\) while the proof of 2.\ is by induction on the local height of the proof.
\end{proof}

\begin{figure}
  \[
    \AxiomC{\(\Gamma \gg \Delta\)}
    \RightLabel{\(\wk\)}
    \UnaryInfC{\(\Gamma, \Gamma' \gg \Delta, \Delta'\)}
    \DisplayProof
    \qquad
    \AxiomC{\(\Gamma, \Gamma, \Gamma' \gg \Delta, \Delta, \Delta'\)}
    \RightLabel{\(\ctr\)}
    \UnaryInfC{\(\Gamma, \Gamma' \gg \Delta, \Delta'\)}
    \DisplayProof
    \qquad
    \AxiomC{\(\Gamma \gg \Delta, \chi\)}
    \AxiomC{\(\chi, \Gamma \gg \Delta\)}
    \RightLabel{\(\cut\)}
    \BinaryInfC{\(\Gamma \gg \Delta\)}
    \DisplayProof
  \]
  \caption{Structural Rules}
  \label{fig:structural-rules}
\end{figure}

In Figure~\ref{fig:structural-rules} we can see some common structural rules.
As usual, these rules behave well with the sequent calculi, as we record in the following lemma.
The proof, which is omitted due to constraints of space, is just using induction on the local height and using Theorem~\ref{th:local-adm-and-eliminability} when needed.

\begin{lemma}
  In \(\g{\GLS}\; (+\cut)\) and in \(\n{\GLS} (+\cut)\) we have that
  \begin{enumerate}
    \item \(\wk\) is eliminable and admissible preserving local height and local rules,
    \item \((\botR)\), \((\toL)\), \((\toR)\) are invertible preserving local height and local rules,
    \item \(\ctr\) is admissible preserving local height and local rules.
  \end{enumerate}
\end{lemma}

Using the admissibility of structural rules is easy to show the following (see \cite{GLSprooftheory}).

\begin{theorem}
  For any multisets \(\Gamma, \Delta\) we have that
  \[
    \g{\GLS} + \cut \vdash \Gamma \Rightarrow \Delta \text{ iff }\GL \vdash \bigwedge \Gamma \to \bigvee \Delta 
    \quad \text{and}\quad
    \g{\GLS} + \cut \vdash \Gamma \Rrightarrow \Delta \text{ iff }\GLS \vdash \bigwedge \Gamma \to \bigvee \Delta.
  \]
\end{theorem}

\subsection{Translations between wellfounded and non-wellfounded proofs}

To establish cut elimination for \(\g{\GLS}\) via cut elimination for \(\n{\GLS}\), we will need to be able to translate proofs between both calculi.
In this section we define the necessary translations.
The following lemma follows from using \((\cut)\) and \((\modal{\GL})\).

\begin{lemma}[L\"ob's rule]
  The following rule is admissible in \(\g{\GLS} + \cut\):
  \[
    \AxiomC{\(\necd \Sigma, \nec\phi \Rightarrow \phi\)}
    \RightLabel{\(\lob\)}
    \UnaryInfC{\(\necd \Sigma \Rightarrow \phi\)}
    \DisplayProof
  \]
\end{lemma}
%\begin{proof}
%  Let \(\pi \vdash \necd \Sigma, \nec \phi \Rightarrow \phi\) in \(\g{\GLS} + \cut\), we have the following proof
%  \[
%    \AxiomC{\(\pi\)}
%    \noLine
%    \UnaryInfC{\(\necd \Sigma, \nec \phi \Rightarrow \phi\)}
%    \RightLabel{\(\modal{\GL}\)}
%    \UnaryInfC{\(\necd \Sigma \Rightarrow \phi, \nec\phi\)}
%    \AxiomC{\(\pi\)}
%    \noLine
%    \UnaryInfC{\(\necd \Sigma, \nec \phi \Rightarrow \phi\)}
%    \RightLabel{\(\cut\)}
%    \BinaryInfC{\(\necd \Sigma \Rightarrow \phi\)}
%    \DisplayProof
%  \]
%\end{proof}

\begin{lemma}\label{lm:from-wellfounded-to-nonwellfounded}
  If\/ \(\g{\GLS} + \cut \vdash \Gamma \Rightarrow \Delta\), then \(\n{\GLS} + \cut \vdash \Gamma \Rightarrow \Delta\).
\end{lemma}
\begin{proof}
  We define corecursively a function \(\alpha\) from proofs in \(\g{\GLS} + \cut\) to proofs in \(\n{\GLS} + \cut\) by cases on the last rule of the input proof. 
  It commutes with all rules different from \((\modal{\GL})\), and for \((\modal{\GL})\) it is defined as
  \[
    \AxiomC{\(\pi_0\)}
    \noLine
    \UnaryInfC{\(\necd \Sigma, \nec \phi \Rightarrow \phi\)}
    \RightLabel{\(\modal{\GL}\)}
    \UnaryInfC{\(\nec \Sigma, \Gamma \Rightarrow \nec\phi, \Delta\)}
    \DisplayProof
    \quad
    \overset{\alpha}{\longmapsto}
    \quad
    \AxiomC{\(\alpha(\lob(\pi_0))\)}
    \noLine
    \UnaryInfC{\(\necd \Sigma\Rightarrow \phi\)}
    \RightLabel{\(\modal{\KT}\)}
    \UnaryInfC{\(\nec \Sigma, \Gamma \Rightarrow \nec\phi, \Delta\)}
    \DisplayProof
  \]
  %and
  %\[
  %  \AxiomC{\(\pi_0\)}
  %  \noLine
  %  \UnaryInfC{\(S_0\)}
  %  \AxiomC{\(\cdots\)}
  %  \AxiomC{\(\pi_{n-1}\)}
  %  \noLine
  %  \UnaryInfC{\(S_{n-1}\)}
  %  \RightLabel{\(R\)}
  %  \TrinaryInfC{\(S\)}
  %  \DisplayProof
  %  \quad
  %  \overset{\alpha}{\longmapsto}
  %  \quad
  %  \AxiomC{\(\alpha(\pi_0)\)}
  %  \noLine
  %  \UnaryInfC{\(S_0\)}
  %  \AxiomC{\(\cdots\)}
  %  \AxiomC{\(\alpha(\pi_{n-1})\)}
  %  \noLine
  %  \UnaryInfC{\(S_{n-1}\)}
  %  \RightLabel{\(R\)}
  %  \TrinaryInfC{\(S\)}
  %  \DisplayProof
  %\]
  %for \(R \neq \modal{\GL}\).
  We notice that \(\alpha(\pi)\) is indeed a proof and not a preproof, the corecursive call only increase in height if progress is made in the path from the root to the corecursive call.
\end{proof}

\begin{definition}
  Let \(\phi\) be a formula.
  We define its \emph{set of subformulas}, denoted \(\sub(\phi)\), recursively on \(\phi\) as
  \begin{align*}
    &\sub(p) = \set{p},
    &&\sub(\bot) = \set{\bot}, \\
    &\sub(\phi_0 \to \phi_1) = \set{\phi_0 \to \phi_1} \union \sub(\phi_0) \union \sub(\phi_1),
    &&\sub(\nec \phi_0) = \set{\nec \phi_0} \union \sub(\phi_0).
  \end{align*}
  Given a multiset \(\Gamma\) we will write \(\sub(\Gamma)\) to denote the set \(\Union_{\phi \in \Gamma} \sub(\phi)\), and given a sequent \(\Gamma \gg \Delta\) we will write \(\sub(\Gamma \gg \Delta)\) to denote \(\sub(\Gamma) \union \sub(\Delta)\).
\end{definition}

The following lemma follows by inspection of Figure~\ref{fig:sequent-rules}.

\begin{lemma}[Local subformula property]
  Let
  \[
    \AxiomC{\(S_0\)}
    \AxiomC{\(\cdots\)}
    \AxiomC{\(S_{n-1}\)}
    \RightLabel{\(R\)}
    \TrinaryInfC{\(S\)}
    \DisplayProof
  \]
  be an instance of a rule of\/ \(\g{\GLS}\) or \(\n{\GLS}\).
  Then \(\sub(S_i) \subseteq \sub(S)\) for \(i < n\).
\end{lemma}

\begin{lemma}\label{lm:from-nonwellfounded-to-wellfounded}
  For any finite set \(\Lambda\), we have that \(\n{\GLS} \vdash \Gamma \gg \Delta\) implies \(\g{\GLS} \vdash \nec \Lambda, \Gamma \gg \Delta\).
\end{lemma}
\begin{proof}
  Let \(\pi \vdash \Gamma \gg \Delta\) in \(\n{\GLS}\),
  we proceed by induction on the measure \(\omega\cdot | \sub(\Gamma \gg \Delta) \setminus \Lambda| + \lhg(\pi)\) and cases on the last rule of \(\pi\).
  If the last rule of \(\pi\) is not \(\modal{\KT}\) then it suffices to apply the induction hypothesis and reapply the rule.
  Now assume that the last rule of \(\pi\) is \(\modal{\KT}\), i.e., \(\pi\) is
  \[
    \AxiomC{\(\pi_0\)}
    \noLine
    \UnaryInfC{\(\necd \Sigma \Rightarrow \phi\)}
    \RightLabel{\(\modal{\KT}\)}
    \UnaryInfC{\(\nec \Sigma, \Gamma' \gg \nec\phi, \Delta'\)}
    \DisplayProof
  \]
  where \(\Gamma = \nec \Sigma, \Gamma'\) and \(\Delta = \nec \phi, \Delta'\).
  Let us denote the conclusion of \(\pi_0\) as \(S_0\) and the conclusion of \(\pi\) as~\(S\).
  We proceed by cases.  If \(\phi \in \Lambda\) the desired proof is obtained by using \((\Ax)\), as \(\nec \phi\) will be in both sides of the sequent.
  Now assume that \(\phi \not \in \Lambda\), then
      \(
        |\sub(S_0) \setminus (\Lambda \union \set{\phi})| < |\sub(S_0) \setminus \Lambda| \leq |\sub(S) \setminus \Lambda|,
      \)
      so by the induction hypothesis we have a proof \(\tau \vdash \nec \phi, \nec \Lambda, \necd \Sigma \Rightarrow \phi\).
      The desired proof is then
      \[
        \AxiomC{\(\wk(\tau)\)}
        \noLine
        \UnaryInfC{\(\nec \phi, \necd \Lambda, \necd \Sigma \Rightarrow \phi\)}
        \RightLabel{\(\modal{\GL}\)}
        \UnaryInfC{\(\nec \Lambda, \nec \Sigma, \Gamma' \Rightarrow \nec\phi, \Delta'\)}
        \DisplayProof
      \]
\end{proof}

\subsection{Cut elimination}

Finally, we show local admissibility of \(\cut\) in \(\n{\GLS}\) which implies its eliminability.
Thanks to the translations of the previous subsection we also obtain cut elimination for \(\g{\GLS}\).

\begin{theorem}\label{th:cut-elimination-gls}
  \(\cut\) is locally admissible in \(\n{\GLS}\).
  As a corollary, \(\cut\) is eliminable in \(\n{\GLS}\).
\end{theorem}
\begin{proof}
  Let \(\pi \vdash \Gamma \gg \Delta,\chi\) and \(\tau \vdash \chi, \Gamma \gg \Delta\) be locally cut free proofs in \(\n{\GLS} + \cut\).
  We proceed by induction on the measure \(\omega \cdot |\chi| + (\lhg(\pi) + \lhg(\tau))\).
  Most cases follow the usual cut reductions in cut elimination (for the details see Appendix~\ref{sec:cut-reductions}), so we will consider only two special cases.
  Assume that \(\pi\) and~\(\tau\) are of the following shape
  \[
    \AxiomC{\(\pi_0\)}
    \noLine
    \UnaryInfC{\(\necd \Sigma \Rightarrow \chi\)}
    \RightLabel{\(\modal{\KT}\)}
    \UnaryInfC{\(\nec \Sigma, \Gamma' \gg \nec\chi, \nec \phi, \Delta'\)}
    \DisplayProof
    \qquad
    \AxiomC{\(\tau_0\)}
    \noLine
    \UnaryInfC{\( \necd \chi, \necd \Sigma \Rightarrow \phi\)}
    \RightLabel{\(\modal{\KT}\)}
    \UnaryInfC{\(\nec\chi, \nec \Sigma, \Gamma' \gg \nec \phi, \Delta'\)}
    \DisplayProof
  \]
  where \(\Gamma = \nec \Sigma, \Gamma'\), \(\Delta = \nec \phi, \Delta'\).
  We can assume that the auxiliary \(\nec\)-formulas on the left hand side are the same, since otherwise we can use \(\wk\) on \(\pi_0\) and \(\tau_0\) preserving height and local cut freeness.
  Then, the desired proof is
  \[
    \AxiomC{\(\wk(\pi_0)\)}
    \noLine
    \UnaryInfC{\(\necd \Sigma \Rightarrow \phi, \chi\)}
    \AxiomC{\(\pi_0\)}
    \noLine
    \UnaryInfC{\(\necd \Sigma \Rightarrow \chi\)}
    \RightLabel{\(\modal{\KT}\)}
    \UnaryInfC{\(\chi, \necd \Sigma \Rightarrow \phi, \nec \chi\)}
    \AxiomC{\(\tau_0\)}
    \noLine
    \UnaryInfC{\(\necd \chi, \necd \Sigma \Rightarrow \phi\)}
    \RightLabel{\(\cut\)}
    \BinaryInfC{\(\chi, \necd \Sigma \Rightarrow \phi\)}
    \RightLabel{\(\cut\)}
    \BinaryInfC{\(\necd \Sigma \Rightarrow \phi\)}
    \RightLabel{\(\modal{\KT}\)}
    \UnaryInfC{\(\nec \Sigma, \Gamma' \Rightarrow \nec\phi, \Delta'\)}
    \DisplayProof
  \]
  Notice that the proof is trivially locally cut free.

  The other case is as follows.
  Assume \(\pi\) and \(\tau\) are of the following shape.
  \[
    \AxiomC{\(\pi_0\)}
    \noLine
    \UnaryInfC{\(\necd \Sigma \Rightarrow \chi\)}
    \RightLabel{\(\modal{\KT}\)}
    \UnaryInfC{\(\nec \Sigma, \Gamma' \Rrightarrow \nec\chi,\Delta\)}
    \DisplayProof
    \qquad
    \AxiomC{\(\tau_0\)}
    \noLine
    \UnaryInfC{\(\necd\chi, \nec \Sigma, \Gamma' \Rrightarrow \Delta\)}
    \RightLabel{\(\S\)}
    \UnaryInfC{\(\nec\chi, \nec \Sigma, \Gamma' \Rrightarrow \Delta\)}
    \DisplayProof
  \]
  where \(\Gamma = \nec \Sigma, \Gamma'\).
  Then the desired proof is
  \[
    \AxiomC{\(\wk(\mathrm{Prom}(\pi_0))\)}
    \noLine
    \UnaryInfC{\(\necd \Sigma,\Gamma' \Rrightarrow \chi, \Delta\)}
    \AxiomC{\(\pi_0\)}
    \noLine
    \UnaryInfC{\(\necd \Sigma \Rightarrow \chi\)}
    \RightLabel{\(\modal{\KT}\)}
    \UnaryInfC{\(\chi, \necd \Sigma, \Gamma' \Rrightarrow \nec\chi \Delta\)}
    \AxiomC{\(\wk(\tau_0)\)}
    \noLine
    \UnaryInfC{\(\necd\chi, \necd \Sigma, \Gamma' \Rrightarrow \Delta\)}
    \RightLabel{\(\cut\) (I.H.)}
    \BinaryInfC{\(\chi, \necd \Sigma, \Gamma' \Rrightarrow \Delta\)}
    \RightLabel{\(\cut\) (I.H.)}
    \BinaryInfC{\(\necd \Sigma, \Gamma' \Rrightarrow \Delta\)}
    \doubleLine
    \RightLabel{\(\S\)}
    \UnaryInfC{\(\nec \Sigma, \Gamma' \Rrightarrow \Delta\)}
    \DisplayProof
  \]
  where the double line denotes muliple aplpications of \((\S)\).
  The first cut is justified by a smaller sum of local heights and the second cut by a smaller cut formula.
\end{proof}

Thanks to the translations of the previous subsection, we obtain the following.

\begin{corollary}
  \(\cut\) is eliminable in \(\g{\GLS}\).\footnote{
In fact, this cut elimination procedure can be made effective (see  the footnote on page~\pageref{def:prop-of-rules}) since a finite approximation of the non-wellfounded proof will suffice to compute each of the steps. }
\end{corollary}

\begin{corollary}
  Let \(\Gamma, \Delta\) be multisets.
  We have the following:
  \begin{enumerate}
    \item  \(\GL \vdash \bigwedge \Gamma \to \bigvee \Delta \quad\text{iff}\quad \g{\GLS} \;(+\cut) \vdash \Gamma \Rightarrow \Delta \quad\text{iff}\quad \n{\GLS} \;(+\cut)\vdash \Gamma \Rightarrow \Delta\).
    \item  \(\GLS \vdash \bigwedge \Gamma \to \bigvee \Delta \quad\text{iff}\quad \g{\GLS} \;(+\cut)\vdash \Gamma \Rrightarrow \Delta \quad\text{iff}\quad \n{\GLS} \;(+\cut)\vdash \Gamma \Rrightarrow \Delta\).
  \end{enumerate}
\end{corollary}

\section{Uniform Lyndon Interpolation for \(\GLS\)}\label{sec:ulip}

Having established the connection between the sequent and Hilbert calculi for \(\GLS\), we are ready to show uniform Lyndon interpolation for \(\GLS\).
In order to respect the polarity of variables, it is necessary to avoid introducing the diagonal formula in the modal rule. 
Hence the use of \(\n{\GLS}\) instead of \(\g{\GLS}\) is essential.

To properly talk about uniform Lyndon interpolation, first we need to define the positive and negative vocabulary of a formula. From now on, with \emph{vocabulary} we simply mean a set of variables.

\begin{definition}
  Let \(\phi\) be a formula.
  We define its \emph{positive} and \emph{negative vocabulary}, denoted \(\voc_+(\phi)\) and \(\voc_-(\phi)\) respectively, recursively by
  \begin{align*}
    &\voc_+(p) = \set{p},
    &&\voc_-(p) = \varnothing, \\
    &\voc_+(\bot) = \varnothing,
    &&\voc_-(\bot) = \varnothing, \\
    &\voc_+(\phi_0 \to \phi_1) = \voc_-(\phi_0) \union \voc_+(\phi_1),
    &&\voc_-(\phi_0 \to \phi_1) = \voc_+(\phi_0) \union \voc_-(\phi_1), \\
    &\voc_+(\nec\phi_0) = \voc_+(\phi_0),
    &&\voc_-(\nec\phi_0) = \voc_-(\phi_0).
  \end{align*}
  Given a multiset \(\Gamma\) we will write \(\voc_+(\Gamma)\) to mean the set \(\Union_{\phi \in \Gamma} \voc_+(\phi)\), and given a sequent \(\Gamma \gg \Delta\) we will write \(\voc_+(\Gamma \gg \Delta)\) to mean \(\voc_+(\Gamma) \union \voc_+(\Delta)\).
  We will use the same notation for \(\voc_-\).
\end{definition}

We will write \(\overline{+}\) to mean \(-\) and \(\overline{-}\) to mean \(+\).
The following lemma can be proved by induction on the complexity of \(\phi\) (although we omit the proof as it is highly technical and not particularly illuminating).

\begin{lemma}
  \label{lm:substitution-and-polarity}
  Let \(p\) be a propositional variable and \(\phi(p_0,\ldots,p_{n-1}, q_0,\ldots,q_{m-1})\), \(\psi_0,\ldots,\psi_{n-1}\),  and \(\chi_0,\ldots,\chi_{m-1}\) be formulas. 
%  We have that 
If\/ \(\voc_-(\phi) \cap \set{p_0,\ldots,p_{n-1}} =\varnothing\) and\/  \(\voc_+(\phi) \cap \set{q_0,\ldots,q_{m-1}} = \varnothing\), then for \(b \in \set{+,-}\)
  \begin{multline*}
    \voc_b(\phi(\psi_0,\ldots,\psi_{n-1}, \chi_0,\ldots,\chi_{m-1})) \subseteq\\
    \voc_b(\phi) \setminus\set{p_0,\ldots,p_{n-1},q_0,\ldots,q_{m-1}}
    \union \Union_{i < n} \voc_b(\psi_i) \union \Union_{j < m} \voc_{\overline{b}}(\chi_j).
  \end{multline*}
\end{lemma}

We are ready to recall the definition of uniform Lyndon interpolation. 

\begin{definition}
  A logic \(L\) has the \emph{Uniform Lyndon Interpolation Property} (or ULIP in short) if for any formula \(\phi\) and vocabularies \(V_+, V_-\), there is a formula \(\iota\) called \emph{uniform Lyndon interpolant} such that
  \begin{enumerate}
    \item \(\voc_+(\iota) \subseteq V_+\) and \(\voc_-(\iota) \subseteq V_-\),
    \item \(L \vdash \phi \to \iota\), and
    \item for any \(\psi\) such that \(\voc_+(\psi) \subseteq V_+\) and \(\voc_-(\psi) \subseteq V_-\) we have that \(L \vdash \phi \to \psi\) implies \(L \vdash \iota \to \psi\).
  \end{enumerate}
\end{definition}
The rest of the paper will be dedicated to show that \(\GLS\) has the ULIP.

\subsection{Lyndon equational systems}

The first step towards calculating uniform Lyndon interpolants will be to solve equational systems of formulas.
In addition, as the polarities of variables matters, we will need to keep track of them.
We will do so by introducing the notions of \emph{Lyndon fixpoint} and \emph{Lyndon equational system}.

\begin{definition}[Lyndon Fixpoints]
  Let \(\phi(p)\) and \(\psi\) be formulas.
  We say that \(\psi\) is a \emph{Lyndon fixpoint of~\(\phi\) (with respect to \(p\)) in \(\GLS\)} if \(\voc_{b}(\psi) \subseteq \voc_b(\phi) \setminus \set{p}\) for \(b \in \set{+,-}\) and \(\GLS \vdash \psi \leftrightarrow \phi(\psi)\).
\end{definition}

A formula \(\phi\) is modalized in a variable \(p\) if all the occurences of \(p\) in \(\phi\) are under the scope of a \(\nec\).

\begin{definition}[Lyndon equational systems]
  Let \(V_+,V_-\) be vocabularies and \(\bar{p} = (p_0,\ldots,p_{n-1})\) be a finite sequence of pairwise different variables not occuring in \(V_+ \union V_-\).
  A \emph{Lyndon equational system over \((\bar{p},V_+,V_-)\)} is a collection of triples \(\mathcal{E} = \set{(p_i,b_i,\phi_i) \mid i < n}\) such that \(b_i \in \set{+,-}\), \(\phi_i\) is a formula, \(\voc_{b_i}(\phi_i) \subseteq V_+ \union B_+\) and \(\voc_{\overline{b_i}}(\phi_i) \subseteq V_- \union B_-\), where \(B_+ = \set{p_i \mid b_i = +}\) and \(B_- = \set{p_i \mid b_i = -}\).
  The elements of \(\bar{p}\) are called \emph{unknowns} and those of \(\mathcal{E}\) are called \emph{equations}.

  A \emph{solution in \(L\)} to \(\mathcal{E}\) is a sequence  \((\psi_0,\ldots,\psi_{n-1})\) such that for each \(i \in n\),  we have
  \(\voc_{b_i}(\psi_i) \subseteq V_+\), \(\voc_{\overline{b_i}}(\psi_i) \subseteq V_-\) and \(L \vdash \psi_i \leftrightarrow \phi_i[\psi_0/p_0,\ldots,\psi_{n-1}/p_{n-1}]\).
  \(\mathcal{E}\) is said to be
  \begin{multicols}{2}
    \begin{enumerate}
      \item \emph{Solvable in \(L\)} if it has a solution in \(L\).
      \item \emph{Simple} if \(\phi_i\) is a \(\nec\)-formula for \(i < n\).
      \item \emph{Modalized} if \(\phi_{i}\) is modalized in \(p_0,\ldots,p_i\).
      \item \emph{Positive} if \(b_i = {+}\) for  \(i < n\).
    \end{enumerate}
  \end{multicols}
\end{definition}

Given an equational system \(\mathcal{E}\) over \((\bar{p}, V_+, V_-)\) with solution \((\psi_0,\ldots,\psi_{n-1})\) we will also call the substitution \((\cdot)^*\) where \(p_i^* = \psi_i\) and \(q^* = q\) for \(q\) not in \(\bar{p}\) a solution of \(\mathcal{E}\).

\begin{lemma}\label{lm:simple-equation-systems}
  We have that
  \begin{enumerate}
    \item For any formula \(\phi(p)\), \(\nec \phi(\top)\) is a Lyndon fixpoint of\/ \(\nec \phi(p)\) in \(\GLS\).
    \item Simple Lyndon equational systems have a solution in \(\GLS\).
  \end{enumerate}
\end{lemma}
\begin{proof}
  Proof 1. It is wellknown (e.g.~see \cite[p.~78]{smorynski}) that \(\GL \vdash \nec \phi(\top) \leftrightarrow \nec \phi(\nec \phi (\top))\), so it is also a theorem of \(\GLS\).
  By the properties of substitution, we get  \(\voc_{b}(\phi(\top)) \subseteq \voc_{b}(\phi) \setminus \set{p}\) for \(b \in \set{+,-}\).

  Proof of 2.
  Let \(\mathcal{E} = \set{(p_i,b_i,\nec \phi_i) \mid i < n}\) be a simple Lyndon equational system over \((\bar{p}, V_+, V_-)\).
  We proceed by induction on \(n\), the number of unknowns.
  Take \(\nec \phi_0(p_0,\ldots,p_{n-1})\), we know that it has a Lyndon fixpoint \(\psi_0\) in \(\GLS\).
  Note that \(\mathcal{E}' = \set{(p_i, b_i, \nec \phi_i[\psi_0/p_0]) \mid 1 \leq i < n}\) is a simple Lyndon equational system solvable in \(\GLS\) by the induction hypothesis, let \((\chi_1,\ldots,\chi_n)\) be a solution in \(\GLS\) of it.
  Let us define \(\chi_0 = \psi_0[\chi_1/p_1,\ldots,\chi_n/p_n]\), then we claim that \((\chi_0,\ldots,\chi_n)\) is a solution of \(\mathcal{E}\).

  First, note that we already have that \(\voc_{b_i}(\chi_i) \subseteq V_+\) and \(\voc_{\overline{b_i}}(\chi_i) \subseteq V_-\) for \(1 \leq i < n\).
  We notice that for \(1 \leq i < n\) we have that if \(b_i = b_0\) then \(p_i \not\in \voc_-(\psi_0)\) and if \(b_i \neq b_0\) then \(p_i \not \in \voc_+(\psi_0)\) (it suffices to do cases on \(b_i\) and \(b_0\)).
  Using Lemma~\ref{lm:substitution-and-polarity} we have that
  \begin{align*}
    \voc_{b_0}(\chi_0) &\subseteq \voc_{b_0}(\psi_0) \setminus \set{p_1,\ldots,p_n} \union \Union_{\scriptstyle \begin{matrix} 1 \leq i < n \\ b_i = b_0 \end{matrix}} \voc_{b_0}(\chi_i) \union \Union_{\scriptstyle \begin{matrix} 1 \leq i < n \\ b_i \neq b_0 \end{matrix}} \voc_{\overline{b_0}}(\chi_i) \\
    &\subseteq \voc_{b_0}(\nec \phi_0) \setminus \set{p_0,\ldots,p_n} \union \Union_{\scriptstyle \begin{matrix} 1 \leq i < n \\ b_i = b_0 \end{matrix}} \voc_{b_i}(\chi_i) \union \Union_{\scriptstyle \begin{matrix} 1 \leq i < n \\ b_i \neq b_0 \end{matrix}} \voc_{b_i}(\chi_i) \subseteq V_+.
  \end{align*}
  where we used that \(\voc_{b_0}(\nec \phi_0) \subseteq V_+ \union B_+\).
  We have an anologous reasoning for \(\voc_{\overline{b_0}}(\chi_0)\), so the desired polarity conditions hold.
  
  Finally, we show the desired equivalences.
  We have \(\GLS \vdash \chi_i \leftrightarrow (\nec \phi_i[\psi_0/p_0])[\chi_1/p_1, \ldots, \chi_{n}/p_n]\) for \(1 \leq i < n\).
  Using that \(p_0 \neq p_i\) for \(1 \leq i < n\) we obtain that 
  \begin{align*}
    (\nec\phi_i[\psi_0/p_0])[\chi_1/p_1, \ldots, \chi_{n}/p_n] 
    &= \nec\phi_i[\psi_0[\chi_1/p_1, \ldots, \chi_{n}/p_n]/p_0, \chi_1/p_1, \ldots, \chi_{n}/p_n]\\
    &= \nec\phi_i[\chi_0/p_0,\chi_1/p_1, \ldots, \chi_{n}/p_n],
  \end{align*}
  as desired.
  All left to show that \(\GLS \vdash \chi_0 \leftrightarrow \nec \phi_0[\chi_0/p_0, \ldots, \chi_{n}/p_n]\).
  We have that, as \(\psi_0\) is a fixpoint, \(\GLS \vdash \nec \psi_0 \leftrightarrow \nec \phi_0[\psi_0/p_0]\).
  As \(\GLS\) is closed under substitutions we obtain that  \(\GLS \vdash \nec \chi_0 \leftrightarrow (\nec \phi_0[\psi_0/p_0])[\chi_1/p_1,\ldots,\chi_n/p_n]\), and we can use the same equalities as before since \(p_0 \neq p_i\) for \(1 \leq i < n\).
\end{proof}

We say that \(\phi\) is positive in \(p\) if \(p \not \in \voc_-(\phi)\).

\begin{theorem}\label{th:positive-modalized-equation-systems}
  We have that
  \begin{enumerate}
    \item Every formula positive and modalized in \(p\) has a Lyndon fixpoint in \(\GLS\).
    \item Positive modalized Lyndon equational system are solvable in \(\GLS\).
  \end{enumerate}
\end{theorem}
\begin{proof}
  Proof of 1.\footnote{This proof follows the proof in \cite{lindstrom} for \(\GL\), with the addition of the condition on the polarity of variables.}
  Let \(\phi(p)\) be a formula that is modalized in \(p\) and positive in \(p\).
  Then \(\phi(p)\) can be written as \( \phi'(\nec\psi_0(p),\ldots,\nec\psi_{n-1}(p),\nec\chi_0(p),\ldots,\nec\chi_{m-1}(p))\),
  where \(\phi'(q_0,\ldots,q_{n-1}, r_0,\ldots,r_{m-1})\) does not contain \(p\),
    %the \(\psi_i\)s are different, and so are the \(\chi_i\)s and
    and \(\set{q_0,\ldots,q_{n-1}} \subseteq \voc_+(\phi') \setminus \voc_-(\phi') \), \(\set{r_0,\ldots,r_{m-1}} \subseteq \voc_-(\phi') \setminus \voc_+(\phi')\).
    From this is easy to infer that  for \(i < n\), \(j < m\) and \(b \in \set{+,-}\):
    \[
      \voc_b(\psi_i) \subseteq \voc_b(\phi), \quad \voc_b(\chi_j) \subseteq \voc_{\overline{b}}(\phi), \quad
      p \not \in \voc_-(\psi_i)\quad \text{and}\quad p \not \in \voc_+(\chi_j).
    \]
  %  In addition, note the following facts following from these condtions. 
  %\begin{itemize}
  %  \item \(\voc_+(\phi') \subseteq \voc_+(\phi) \setminus\set{p} \union \set{q_0,\ldots,q_{n-1}}\) and 
  %  \(\voc_-(\phi') \subseteq \voc_-(\phi) \setminus\set{p} \union \set{r_0,\ldots,r_{m-1}}\).
  %  \item The last condition implies that \(\voc_+(\psi_i) \subseteq \voc_+(\phi)\), \(\voc_-(\psi_i) \subseteq \voc_-(\phi)\) for \(i < n\) and \(\voc_+(\chi_j) \subseteq \voc_-(\phi)\), \(\voc_-(\chi_j) \subseteq \voc_+(\phi)\) for \(j < m\).
  %  \item The last condition together with the fact that \(p \not \in \voc_-(\phi)\) give us that \(p \not \in \voc_-(\psi_i)\) for \(i < n\) and \(p \not \in \voc_+(\chi_j)\) for \(j < m\).
  %\end{itemize}
  Consider the following equational system
  \( \set{(q_i, +,\nec \psi_i(\phi')) \mid i < n} \union \set{(r_j,-, \nec\chi_j(\phi')) \mid j < m}.
  \)
  This is a simple Lyndon \((\bar{q}\bar{r}, \voc_+(\phi)\setminus \set{p}, \voc_-(\phi)\setminus \set{p})\)-equational system.
  %Let us check this.
  %\begin{align*}
  %  &\voc_+(\psi_i(\phi'))
  %  \subseteq \voc_+(\psi_i) \setminus \set{p} \union \voc_+(\phi')
  %  \subseteq \voc_+(\phi) \setminus \set{p} \union \set{q_0,\ldots,q_{n-1}}, \\
  %  &\voc_-(\psi_i(\phi'))
  %  \subseteq \voc_-(\psi_i) \setminus \set{p} \union \voc_-(\phi')
  %  \subseteq \voc_-(\phi) \setminus \set{p} \union \set{r_0,\ldots,r_{m-1}},\\
  %  &\voc_+(\chi_j(\phi'))
  %  \subseteq \voc_+(\chi_j) \setminus \set{p} \union \voc_-(\phi')
  %  \subseteq \voc_-(\phi) \setminus \set{p} \union \set{r_0,\ldots,r_{m-1}}, \\
  %  &\voc_-(\psi_i(\phi'))
  %  \subseteq \voc_-(\chi_j) \setminus \set{p} \union \voc_+(\phi')
  %  \subseteq \voc_+(\phi) \setminus \set{p} \union \set{q_0,\ldots,q_{n-1}}.
  %\end{align*}
  Thus, it has a solution \((\cdot)^* : \bar{q}\bar{r}  \longrightarrow \mathcal{L}_{\nec}\) in \(\GLS\).
  Define \(\eta = \phi'(q^*_0, \ldots, q^*_{n-1}, r^*_0, \ldots,r^*_{m-1})\), by Lemma~\ref{lm:substitution-and-polarity} we know that \(\voc_+(\eta) \subseteq \voc_+(\phi) \setminus \set{p}\) and \(\voc_-(\eta) \subseteq \voc_-(\phi)\).
  Also, since \((\cdot)^*\) is a solution of the equational system, we have that \(\GLS \vdash q^*_i \leftrightarrow \nec \psi_i(\eta)\) and \(\GLS \vdash r^*_j \leftrightarrow \nec \chi_j(\eta)\).  So we obtain
  \[\GLS \vdash \eta \leftrightarrow \phi'(\nec\psi_0(\eta), \ldots,\nec\psi_{n-1}(\eta), \nec\chi_0(\eta), \ldots,\nec\chi_{m-1}(\eta)). \]
  In other words, \(L \vdash \eta \leftrightarrow \phi(\eta)\), so \(\eta\) is a Lyndon fixpoint of \(\phi\) with respect to \(p\).

  Proof of 2. The proof is similar to the second point of Lemma~\ref{lm:simple-equation-systems} using that if \(\phi_i\) is modalised in \(p_0,\ldots,p_i\) then \(\phi_i[\psi/p_0]\) is also modalised in \(p_0,\ldots,p_i\).
\end{proof}

The fundamental result will be the second part of Theorem~\ref{th:positive-modalized-equation-systems}.
When constructing an interpolant we will see that a positive modalized Lyndon equational system will arise.
Solving it and using the solution on a particular formula will give us the uniform Lyndon interpolant.

\subsection{Interpolation templates}

It is common to relate uniform interpolation to proof search.
We will make this connection explicit via defining interpolation templates, which is a proof search where we assume that we only have a part of the desired sequent.
Since there is a rule in the system which increase the complexity of sequents (namely rule~\((\S)\)), we will need to base our proof search on a notion of saturation.

\begin{definition}
  Let \(\Gamma \gg \Delta\) be a sequent.
  We say that a formula \(\phi\) is
  \begin{multicols}{2}
    \begin{enumerate}
      \item \emph{Left saturated in \(\Gamma \gg \Delta\)} if
        \begin{enumerate}
          \item \(\phi\) is atomic,
          \item \(\phi = \phi_0 \to \phi_1\) and \(\phi_0 \in \Delta\) or \(\phi_1 \in \Gamma\),
          \item \(\phi = \nec \phi_0\) and \(\phi_0 \in \Gamma\) or \({\gg} = { \Rightarrow}\).
        \end{enumerate}
      \item \emph{Right saturated in \(\Gamma \gg \Delta\)} if
        \begin{enumerate}
          \item \(\phi\) is atomic or a \(\nec\)-formula,
          \item \(\phi = \phi_0 \to \phi_1\) and \(\phi_0 \in \Gamma\), \(\phi_1 \in \Delta\).
        \end{enumerate}
    \end{enumerate}
  \end{multicols}
  Given a sequent \(S = \Gamma \gg \Delta\) we define its \emph{saturation complexity}, denoted \(\satc{S}\), as the multiset
  \[
    \set{|\phi| \mid \phi \in \Gamma, \phi \text{ not left saturated in }S} \union \set{|\phi| \mid \phi \in \Delta, \phi \text{ not right saturated in }S}.
  \]
  We assume the saturation complexity comes with the multiset ordering attached to it, so \(\satc{S} < \satc{S'}\) means that \(\satc{S}\) is obtained from \(\satc{S'}\) by the process of taking out at least one element and replace it with a finite number of strictly smaller elements.
  We remember that the multiset ordering is wellfounded.
\end{definition}

\begin{figure}
  \[
    \AxiomC{\(\)}
    \RightLabel{\(\ax\)}
    \UnaryInfC{\(p, \Gamma \gg p, \Delta\)}
    \DisplayProof
    \quad
    \AxiomC{\(\)}
    \RightLabel{\(\botL\)}
    \UnaryInfC{\(\bot, \Gamma \gg \Delta\)}
    \DisplayProof
    \quad
    \AxiomC{\(\)}
    \RightLabel{\(\emp\)}
    \UnaryInfC{\(\gg\)}
    \DisplayProof
    \quad
    \AxiomC{\(\necd \phi, \Gamma \Rrightarrow \Delta\)}
    \RightLabel{\(\Ssat\)}
    \UnaryInfC{\(\nec \phi, \Gamma \Rrightarrow \Delta\)}
    \DisplayProof
  \]
  where in \((\Ssat)\) the formula \(\nec \phi\) is not left saturated in the conclusion.

  \[
    \AxiomC{\(\phi \to \psi, \Gamma \gg \phi, \Delta\)}
    \AxiomC{\(\psi, \phi \to \psi, \Gamma \gg \Delta\)}
    \RightLabel{\(\toLsat\)}
    \BinaryInfC{\(\phi \to \psi, \Gamma \gg \Delta\)}
    \DisplayProof
    \quad
    \AxiomC{\(\phi, \Gamma \gg \psi, \phi \to \psi, \Delta\)}
    \RightLabel{\(\toRsat\)}
    \UnaryInfC{\(\Gamma \gg \phi \to \psi, \Delta\)}
    \DisplayProof
  \]
  where in \((\toLsat)\) the formula \(\phi \to \psi\) is not left saturated in the conclusion and in \((\toRsat)\) the formula \(\phi \to \psi\) is not right saturated in the conclusion. 

  \[
    \AxiomC{\(\necd \Sigma^s \Rightarrow\)}
    \AxiomC{\([\necd \Sigma^s \Rightarrow \phi]_{\phi \in \Theta}\)}
    \RightLabel{\(\interpolation{\KT}\)}
    \BinaryInfC{\(\nec \Sigma, \Gamma \gg \nec \Theta, \Delta\)}
    \DisplayProof
  \]
  where \(\nec \Sigma, \Gamma \gg \nec \Theta, \Delta\) is saturated, \(\Gamma, \Delta\) contain no \(\nec\)-formulas, \((\Gamma \inter \Delta) \inter \text{Var} = \varnothing \) and \(\bot \not \in \Gamma\).
  \caption{Interpolation Template Rules}
  \label{fig:interpolation-template}
\end{figure}

\begin{definition}
  An interpolation template is a cyclic proof constructed in the local progress calculus whose rules are displayed in Figure~\ref{fig:interpolation-template} and progress is only made at the premises of \(\interpolation{\KT}\).
\end{definition}

It is easy to show the following lemma, the details of the proof can be found in the appendix.

\begin{lemma}\label{lm:interpolation-template}
  Every sequent has an interpolation template.
\end{lemma}

We will consider formulas on a expanded set of propositional variables.
For each \(w \in \mathbb{N}^*\) we pick a new variable \(x_w\) different from the variables in \(\text{Var}\) and such that if \(w,v \in \mathbb{N}^*\) and \(w \neq v\) then \(x_w \neq x_v\).
We define \(B = \set{x_w \mid w \in \mathbb{N}^*}\) and for a interpolation template \(T\) we define \(B_T = \set{x_w \mid w \text{ repeat node of } T}\).
The elements of \(B\) will be called \emph{bound variables} and the formulas on this expanded set of variables will be called \emph{pseudoformulas}.
For the rest of the subsection we fix two vocabularies \(V_+\) and \(V_-\) for which we want to calculate the uniform Lyndon interpolant.

\begin{definition}
  Let \(T\) be an interpolation template and \(w\) a node of \(T\), let us denote the sequent at~\(w\) in \(T\) as \(\Gamma_w \gg_w \Delta_w\).
  We will annotate each of the sequents of \(T\) with a pseudoformula \(\kappa\) called the \emph{preinterpolant at \(w\)}, denoted as \(\kappa : \Gamma_w \gg_w \Delta_w\).
  We proceed by recursion on the tree structure of \(T\) as follows.
  \[
    \AxiomC{\(\)}
    \RightLabel{\(\ax\)}
    \UnaryInfC{\(\bot : p, \Gamma \gg p, \Delta\)}
    \DisplayProof
    \quad
    \AxiomC{\(\)}
    \RightLabel{\(\botL\)}
    \UnaryInfC{\(\bot : \bot, \Gamma \gg \Delta\)}
    \DisplayProof
    \quad
    \AxiomC{\(\)}
    \RightLabel{\(\emp\)}
    \UnaryInfC{\(\top : {\gg}\)}
    \DisplayProof
    \quad
    \AxiomC{}
    \RightLabel{\(\rep\)}
    \UnaryInfC{\(x_w : \Gamma_w \gg_w \Delta_w\)}
    \DisplayProof
    \quad
    \AxiomC{\(\kappa : \necd \phi, \Gamma \Rrightarrow \Delta\)}
    \RightLabel{\(\Ssat\)}
    \UnaryInfC{\(\kappa : \nec \phi, \Gamma \Rrightarrow \Delta\)}
    \DisplayProof
  \]
  \[
    \AxiomC{\(\kappa_0 : \phi \to \psi, \Gamma \gg \phi, \Delta\)}
    \AxiomC{\(\kappa_1 : \psi, \phi \to \psi, \Gamma \gg \Delta\)}
    \RightLabel{\(\toLsat\)}
    \BinaryInfC{\(\kappa_0 \vee \kappa_1 : \phi \to \psi, \Gamma \gg \Delta\)}
    \DisplayProof
    \quad
    \AxiomC{\(\kappa : \phi, \Gamma \gg \psi, \phi \to \psi, \Delta\)}
    \RightLabel{\(\toRsat\)}
    \UnaryInfC{\(\kappa : \Gamma \gg \phi \to \psi, \Delta\)}
    \DisplayProof
  \]
  \[
    \AxiomC{\(\kappa^{\nec} : \necd \Sigma \Rightarrow\)}
    \AxiomC{\([\kappa^{\pos}_\phi : \necd \Sigma \Rightarrow \phi]_{\phi \in \Theta}\)}
    \RightLabel{\(\interpolation{\KT}\)}
    \BinaryInfC{\(\nec \kappa^{\nec} \wedge \bigwedge_{\phi \in \Theta} \pos \kappa^{\pos}_\phi \wedge \bigwedge (\Gamma \cap V_+) \wedge \neg (\Delta \cap V_-) : \nec \Sigma, \Gamma \gg \nec \Theta, \Delta\)}
    \DisplayProof
  \]
\end{definition}

As promised, every interpolation template has an associated positive modalized equational system.
The proof of this fact can be found in the appendix.

\begin{lemma}\label{lm:eq-system-of-template}
  For any interpolation template \(T\) and for any \(x_w \in B_T\) let us write \(\kappa_{x_w}\) to denote the preinterpolant of \(T\) at \(w^\circ\).
  There is an enumeration \(\bar{x} = (x_0,\ldots,x_{n-1})\) of \(B_T\) such that the equational system \(\mathcal{E}_T = \set{(x_i, +, \kappa_{x_i}) \mid i < n}\) is a positive modalized equational system over \((\bar{x}, V_+, V_-)\).
  As a corollary, \(\mathcal{E}_T\) is solvable in \(\GLS\).
\end{lemma}

\begin{definition}[Interpolant]
  Let \(T\) be an interpolation template, \((\cdot)^*\) a solution of \(\mathcal{E}_T\) and for each node~\(w\) of \(T\) let us denote the preinterpolant at \(w\) as \(\kappa_w\).
  The interpolant given by \(T\) is defined as \(\iota_T = \kappa_\epsilon^*\).
\end{definition}

The interpolant given by an interpolation template \(T\) depends on the chosen solution of \(\mathcal{E}_T\).
This does not matter to us, as in the end one can show that any two uniform interpolants are logically equivalent.
Also note that, by definition of solution of a Lyndon equational system over \((\bar{x},V_+,V_-)\) we automatically have that each \(\iota_T\) is a formula (and not a pseudoformula) and \(\voc_b(\iota_T) \subseteq V_b\) for \(b \in \set{+,-}\).
Finally, to show that the interpolant has the necessary properties, we will need two additional lemmas.

\begin{lemma}\label{lm:first-verification}
  Let \(T\) be an interpolation template of\/ \(\Gamma \gg \Delta\).
  Then \(\n{\GLS} \vdash \Gamma \gg \Delta, \iota_T\).
\end{lemma}
\begin{proof}
  Given a node \(w\) of \(T\) let \(\Gamma_w \gg_{w} \Delta_w\) denote the sequent at \(w\) and \(\kappa_w\) be the preinterpolant of \(T\) at~\(w\).
  The trick is to build a function \(\alpha\) such that given a node \(w\) of \(T\), \(\alpha(w)\) is a proof in \(\n{\GLS} + \cut + \wk + \ctr\) of \(\Gamma_w \gg_w \Delta_w, \kappa^*_w\).
  Since \(\ctr\) is derivable with \(\cut\), \(\wk\) is eliminable in \(\n{\GLS} + \cut\) and \(\cut\) is eliminable in \(\n{\GLS}\), we can obtain the desired proof in \(\n{\GLS}\).
  \(\alpha\) is defined in Appendix~\ref{sec:proofs-templates}.
\end{proof}

\begin{lemma}\label{lm:second-verification}
  Let \(T\) be an interpolation template of\/ \(\Gamma \gg \Delta\) and let\/  \(\Phi \gg \Psi\) be a sequent such that \(\voc_b(\Phi \gg \Psi) \subseteq V_b\) for \(b \in \set{+,-}\).
  Then
 \(\n{\GLS} \vdash \Gamma, \Gamma' \gg \Delta, \Delta'\) implies \(\n{\GLS} \vdash \iota_T, \Gamma' \gg \Delta'\).
\end{lemma}
\begin{proof}
  Given a node \(w\) of \(T\) let \(\Gamma_w \gg_{w} \Delta_w\) denote the sequent at~\(w\) and \(\kappa_w\) be the preinterpolant of \(T\) at~\(w\).
  The trick is to build a function \(\beta\) such that given a node \(w\) of \(T\) and a proof \(\pi \vdash \Gamma_w, \Phi \gg_{w} \Delta_w, \Psi\) in \(\n{\GLS}\) such that \(\voc_{b}(\Phi \gg_{w} \Psi) \subseteq V_b\) for \(b \in \set{+,-}\) , \(\beta(w,\pi)\) is a proof in \(\n{\GLS} + \cut + \wk + \ctr\) of \(\kappa_w^*, \Phi \gg_w \Psi\).
  Since \(\ctr\) is derivable with \(\cut\), \(\wk\) is eliminable in \(\n{\GLS} + \cut\) and \(\cut\) is eliminable in \(\n{\GLS}\), we can obtain the desired proof in \(\n{\GLS}\).
  The definition of \(\beta\) can be found in Appendix~\ref{sec:proofs-templates}.
\end{proof}

We finish the paper with the desired result.

\begin{theorem}
  \(\GLS\) has uniform Lyndon interpolation.
\end{theorem}
\begin{proof}
  Let \(T\) be an interpolation template for \(\phi \Rrightarrow\).  It is easy to show using Lemmas~\ref{lm:first-verification} and \ref{lm:second-verification} that~\(\iota_T\) is the uniform Lyndon interpolant of \(\phi\) for given vocabularies \((V^+,V_-)\).
\end{proof}

\subsection*{Future work}

The method  to establish uniform Lyndon  interpolation via non-wellfounded proofs can be applied to a plethora of provability logics, once sequent calculi for these logics have been developed.
We leave it as future work to apply the methodology to other provability logics, such as unary interpretability or bimodal provability logics.
In addition, there are two possible extensions of the method which need further exploration.
One direction is to investigate logics without simple Lyndon fixpoints, such as interpretability logic \(\IL\), and see if it is still possible to solve positive modalized Lyndon equational systems.
The other direction is to study whether the method is applicable to intuitionistic provability logics, where the resulting Lyndon equational system does not need to be positive.

\subsection*{Acknowledgements}
We would like to thank Taishi Kurahashi for his interest in these methods and for bringing our attention to the problem of proving ULIP for \(\GLS\).

\newpage
\appendix
\section{Cut reductions}\label{sec:cut-reductions}
We display some possible cut reductions which were not consider at the proof of Theorem~\ref{th:cut-elimination-gls}.

\textbf{Weakening part.}
If \(\chi\) belongs to the weakening part of the last rule instance of \(\pi\) or \(\tau\) we can delete directly, as weakening parts can be modified arbitrarily.
%Let us assume \(\chi\) occurs in the weakening part of \(\pi\), the case for \(\tau\) is analogous.
%\[
%  \AxiomC{\(\pi\)}
%  \noLine
%  \UnaryInfC{\(\Gamma \gg \Delta, \chi\)}
%  \DisplayProof
%  \quad
%  \AxiomC{\(\tau\)}
%  \noLine
%  \UnaryInfC{\(\chi,\Gamma \gg \Delta\)}
%  \DisplayProof
%  \longmapsto
%  \AxiomC{\(\pi\)}
%  \noLine
%  \UnaryInfC{\(\Gamma \gg \Delta\)}
%  \DisplayProof
%\]
%as the weakening part of a rule instance can be changed and still be an instance of the same rule.
From now on we assume that \(\chi\) does not occur in the weakening part of the rule instances.

\textbf{Axiomatic}.
Assume \(\pi\) ends in \((\ax)\), the case for \(\tau\) is analogous.
Then \(\chi = p\) for some variable \(p\) and the desired reduction is
\[
  \AxiomC{\(\)}
  \RightLabel{\(\ax\)}
  \UnaryInfC{\(p, \Gamma' \gg \Delta, p\)}
  \DisplayProof
  \quad
  \AxiomC{\(\tau\)}
  \noLine
  \UnaryInfC{\(p, p, \Gamma' \gg \Delta\)}
  \DisplayProof
  \longmapsto
  \AxiomC{\(\ctr(\tau)\)}
  \noLine
  \UnaryInfC{\(p, \Gamma' \gg \Delta\)}
  \DisplayProof
\]
where \(\Gamma = p, \Gamma'\).

If \(\pi\) ends in \(\botL\) then \(\chi\) would belong to the weakening part, which is already covered.
Assume \(\tau\) ends in \(\botL\), so \(\chi = \bot\).
\[
  \AxiomC{\(\pi\)}
  \noLine
  \UnaryInfC{\(\Gamma \gg \Delta, \bot\)}
  \DisplayProof
  \quad
  \AxiomC{\(\tau\)}
  \noLine
  \UnaryInfC{\(\bot, \Gamma \gg \Delta\)}
  \DisplayProof
  \longmapsto
  \AxiomC{\(\mathrm{inv}_{\botR}(\pi)\)}
  \noLine
  \UnaryInfC{\(\Gamma \gg \Delta\)}
  \DisplayProof
\]

From now own we assume that neither \(\pi\) nor \(\tau\) end in \((\ax)\) or \((\botL)\).

\textbf{\(\botR\) case}.
Asssume \(\pi\) ends in an application of \((\botR)\), the case for \(\tau\) is analogous.
If \(\bot\) is the cut formula, the desired cut reduction is obtained taking the immediate subproof of \(\pi\).
%\[
%  \AxiomC{\(\pi_0\)}
%  \noLine
%  \UnaryInfC{\(\Gamma \gg \Delta\)}
%  \RightLabel{\(\botR\)}
%  \UnaryInfC{\(\Gamma \gg \Delta, \bot\)}
%  \DisplayProof \quad
%  \AxiomC{\(\tau\)}
%  \noLine
%  \UnaryInfC{\(\bot, \Gamma \gg \Delta\)}
%  \DisplayProof
%  \longmapsto
%  \AxiomC{\(\pi_0\)}
%  \noLine
%  \UnaryInfC{\(\Gamma \gg \Delta\)}
%  \DisplayProof
%\]
If \(\bot\) is not the cut formula, the desired cut reduction is
\[
  \AxiomC{\(\pi_0\)}
  \noLine
  \UnaryInfC{\(\Gamma \gg \Delta', \chi\)}
  \RightLabel{\(\botR\)}
  \UnaryInfC{\(\Gamma \gg \bot, \Delta', \chi\)}
  \DisplayProof \quad
  \AxiomC{\(\tau\)}
  \noLine
  \UnaryInfC{\(\chi, \Gamma \gg \bot, \Delta'\)}
  \DisplayProof
  \longmapsto
  \AxiomC{\(\pi_0\)}
  \noLine
  \UnaryInfC{\(\Gamma \gg \Delta', \chi\)}
  \AxiomC{\(\mathrm{inv}_{\botR}(\tau)\)}
  \noLine
  \UnaryInfC{\(\chi, \Gamma \gg \Delta'\)}
  \RightLabel{\(\cut \text{(I.H.)}\)}
  \BinaryInfC{\(\Gamma \gg \Delta'\)}
  \DisplayProof
\]
where \(\Delta = \bot, \Delta'\).

From now own we assume that neither \(\pi\) nor \(\tau\) end in \((\botR)\).

\textbf{Principal cut reduction.}
Assume \(\chi\) is principal in \(\pi\) and \(\tau\).
Then either \(\chi\) is an implication or a \(\nec\)-formula.
The second case is covered at the proof of Theorem~\ref{th:cut-elimination-gls}, so we just show the cut reduction for the first.
\begin{multline*}
  \AxiomC{\(\pi_0\)}
  \noLine
  \UnaryInfC{\(\chi_0, \Gamma \gg \Delta, \chi_1\)}
  \RightLabel{\(\toR\)}
  \UnaryInfC{\(\Gamma \gg \Delta, \chi_0 \to \chi_1\)}
  \DisplayProof
  \quad
  \AxiomC{\(\tau_0\)}
  \noLine
  \UnaryInfC{\(\Gamma \gg \Delta, \chi_0\)}
  \AxiomC{\(\tau_1\)}
  \noLine
  \UnaryInfC{\(\chi_1, \Gamma \gg \Delta\)}
  \RightLabel{\(\toL\)}
  \BinaryInfC{\(\chi_0 \to \chi_1, \Gamma \gg \Delta\)}
  \DisplayProof\\
  \\
  \longmapsto
  \AxiomC{\(\wk(\tau_0)\)}
  \noLine
  \UnaryInfC{\(\Gamma \gg \Delta, \chi_1, \chi_0\)}
  \AxiomC{\(\pi_0\)}
  \noLine
  \UnaryInfC{\(\chi_0, \Gamma \gg \Delta, \chi_1\)}
  \RightLabel{\(\cut \text{(I.H.)}\)}
  \BinaryInfC{\(\Gamma \gg \Delta, \chi_1\)}
  \AxiomC{\(\tau_1\)}
  \noLine
  \UnaryInfC{\(\chi_1, \Gamma \gg \Delta\)}
  \RightLabel{\(\cut \text{(I.H.)}\)}
  \BinaryInfC{\(\Gamma \gg \Delta\)}
  \DisplayProof
\end{multline*}

\textbf{Commutative cut reduction.}
Finally, assume that the cut formula is not principal in either \(\pi\) or \(\tau\).
Assume that the cut formula is not principal in \(\pi\) and the last rule of \(\pi\) is \((\toR)\), the cases where the rule is instead \((\toR)\) or \((\S)\), or where any of these occurs at \(\tau\) are analogous.
The desired cut reduction is

\[
  \AxiomC{\(\pi_0\)}
  \noLine
  \UnaryInfC{\(\phi, \Gamma \Rightarrow \psi, \Delta', \chi\)}
  \RightLabel{\(\toR\)}
  \UnaryInfC{\(\Gamma \Rightarrow \phi \to \psi, \Delta', \chi\)}
  \DisplayProof
  \quad
  \AxiomC{\(\tau\)}
  \noLine
  \UnaryInfC{\(\chi, \Gamma \Rightarrow \phi \to \psi, \Delta'\)}
  \DisplayProof
  \longmapsto
  \AxiomC{\(\pi_0\)}
  \noLine
  \UnaryInfC{\(\phi, \Gamma \Rightarrow \psi, \Delta', \chi\)}
  \AxiomC{\(\mathrm{inv}_{\toR}(\tau)\)}
  \noLine
  \UnaryInfC{\(\chi, \phi, \Gamma \Rightarrow \psi, \Delta'\)}
  \RightLabel{\(\cut\text{ (I.H.)}\)}
  \BinaryInfC{\(\phi, \Gamma \Rightarrow \psi, \Delta'\)}
  \RightLabel{\(\toR\)}
  \UnaryInfC{\( \Gamma \Rightarrow \phi \to \psi, \Delta'\)}
  \DisplayProof
\]
where \(\Delta = \phi \to \psi, \Delta'\)

%If the cut formula is not principal at \(\tau\) and the last rule of \(\tau\) is \((\S)\) or the cut formula is not principal at \(\pi\) and the last rule of \(\pi\) is \((\S)\), the reasoning similar to the previous case using \((\wk)\) instead of \(\mathrm{inv}\)(as inverting \((\S)\) is just using weakening), cutting and reapplying \((\S)\).

Finally, we have that the cut formula is not principal in either \(\pi\) or \(\tau\) with last rule \((\modal{\KT})\).
We notice that this cannot occur in \(\pi\), as then the cut formula would belong to the weakening part, so it must occur in \(\tau\).
To make the cut formula not belong to the weakening part it must be the case that \(\chi = \nec \chi_0\).
If \(\chi\) where not principal in \(\pi\), then we would be in one of the cases which we already covered, so we can assume it is prinicipal in \(\pi\), so \(\pi\) must end in \((\modal{\KT})\).
This case was covered at the proof of Theorem~\ref{th:cut-elimination-gls}.

\section{Proofs about interpolation templates}\label{sec:proofs-templates}

\paragraph{Proof of Lemma~\ref{lm:interpolation-template}.}
We will show that we can obtain a cyclic preproof, this preproof will be a proof since all the rules strictly lower the saturation complexity of sequents from conclusion to premises except for \((\interpolation{\KT})\) (so to repeat a sequent we must at least apply \((\interpolation{\KT})\) once, as otherwise the repeated sequent would have strictly saturation complexity than itself).
Take the sequent \(\Gamma \gg \Delta\) and start applying the rules of Figure~\ref{fig:interpolation-template} as possible (note that for any sequent it is always possible to apply at least one), and whenever a repetition in the same branch is found make the leaf a repetition node.
If we show that in this process no infinite branch is generated we would finish.

Assume otherwise, so there is a branch without repetitions.
As the only rule which not strictly lowers the saturation complexity is \((\interpolation{\KT})\), it must be the case that \((\interpolation{\KT})\) is applied infinitely often.
We notice that the rules of Figure~\ref{fig:interpolation-template} fulfill the subformula property, so all the formulas occuring on the branch must be in the set \(\sub(\Gamma \gg \Delta)\).
Also, the premises of \((\interpolation{\KT})\) are determined by choosing a formula (or nothing in the case of the left-most premise) and a set of formulas, so in fact the possible number of such premises is not bigger than \((k+1)2^k\), where \(k\) is the cardinality of \(\sub(\Gamma \gg \Delta)\).
Since there is infinitely many nodes in the branch which are premises of \((\interpolation{\KT})\), and only finitely many possibilities, there must be a repetition, absurd.

\paragraph{Proof of Lemma~\ref{lm:eq-system-of-template}.}
Given a node \(w\) of \(T\) let \(\hg(w)\) denote the height of \(w\), i.e., the height of the subtree generated by \(w\) (then leafs have height \(0\) and the root has the same height as the tree).
If \(\hg(T) = H\), let \(\bar{x}_i\) for \(i \leq H\) be an arbitrary linear ordering of \(\set{x_{w} \mid w \text{ repeat node and } \hg(w^\circ) = i}\).
Consider the enumeration of \(B_T\) defined as \(\bar{x}_T = \bar{x}_0 \cdots \bar{x}_H\) and let us show that \(\mathcal{E}_T\) is an positive modalized equational system over \((\bar{x}_T, V_+, V_-)\).

The positivy of \(\mathcal{E}_T\) is straightforward by definition.
In addition, one can check that for any node \(w\), \(\voc_+(\kappa_w) \subseteq V_+ \union B_+\) and \(\voc_-(\kappa_w) \subseteq V_-\) by induction on the length of \(w\).
So all left to show is that it is modalized, let \(\bar{x}_T = (x_0,\ldots,x_{n-1})\) and for each \(i < n\) let us write \(w_i\) to mean the repeat node such that \(x_i = x_{w_i}\).
First, we notice three facts. 
\begin{enumerate}
  \item If \(w\) is a repeat node and \(v \not \leq w\) then \(x_w\) does not occur in \(\kappa_w\).
  \item If \(w\) is the conclusion of a \(\modal{\KT}_{\mathrm{sat}}\) rule, then \(\kappa_w\) is modalized in all the variables of \(B_T\).
  \item For any \(x \in B_T\), if \(w\) is a non-repeat node with immediate successors \(w0, \ldots, w(n-1)\) and \(\kappa_{w0}\), \ldots, \(\kappa_{w(n-1)}\) are modalized in \(x\) then \(\kappa_w\) is modalized in \(\kappa_{w}\).
\end{enumerate}
Additionally, thanks to the definition of \(\bar{x}_T\) we have that if \(i \leq j\) then either \(w^\circ_j \leq w^\circ_i\) or \(w^\circ_j\) and \(w^\circ_i\) are incomparable.
Using this with the previous facts gives that for any \(i \leq j\) we have that \(\kappa_{w^\circ_j}\) is modalised in \(x_i\), as desired.

\paragraph{Definition of \(\alpha\) at Lemma~\ref{lm:first-verification}.}
  We will define \(\alpha\) corecursively, that \(\alpha(w)\) is a preproof will be trivial from construction.
  When the definition is finished we will show that it is indeed a proof.
  We proceed by cases on the shape of \(w\).

  Case \(w\) is \((\ax)\) or \((\botL)\).
  \(\alpha(w)\) is obtained by applying rule \((\ax)\) or \((\botL)\), respectively.

  Case \(w\) is \((\emp)\).
  \(\alpha(w)\) is a proof of \(\gg_w \top\), which is straightforward to obtain.

  Case \(w\) is \((\Ssat)\).
  Then the function is defined as
  \[
    \AxiomC{\(w0\)}
    \noLine
    \UnaryInfC{\(\kappa : \necd \phi, \Gamma'_w \Rrightarrow \Delta_w\)}
    \RightLabel{\(\Ssat\)}
    \UnaryInfC{\(\kappa : \nec \phi, \Gamma'_w \Rrightarrow \Delta_w\)}
    \DisplayProof
    \qquad
    \overset{\alpha}{\longmapsto}
    \qquad
    \AxiomC{\(\alpha(w0)\)}
    \noLine
    \UnaryInfC{\(\necd \phi, \Gamma'_w \Rrightarrow \Delta_w, (\kappa)^*\)}
    \RightLabel{\(\S\)}
    \UnaryInfC{\(\nec \phi, \Gamma'_w \Rrightarrow \Delta_w, (\kappa)^*\)}
    \DisplayProof
  \]
  where \(\kappa_w = \kappa\), \({\gg}_w = {\Rrightarrow}\) and \(\Gamma_w = \nec\phi, \Gamma'_w\).

  %Case \(w\) is \((\toRsat)\).
  %Then the function is defined as
  %\[
  %  \AxiomC{\(w0\)}
  %  \noLine
  %  \UnaryInfC{\(\kappa : \phi, \Gamma_w \gg \psi, \phi \to \psi, \Delta'_w\)}
  %  \RightLabel{\(\toRsat\)}
  %  \UnaryInfC{\(\kappa : \Gamma_w \gg \phi \to \psi, \Delta'_w\)}
  %  \DisplayProof
  %  \qquad
  %  \overset{\alpha}{\longmapsto}
  %  \qquad
  %  \AxiomC{\(\alpha(w0)\)}
  %  \noLine
  %  \UnaryInfC{\(\phi, \Gamma_w \gg \psi, \phi \to \psi, \Delta'_w, \kappa^*\)}
  %  \RightLabel{\(\toR\)}
  %  \UnaryInfC{\( \Gamma_w \gg \phi \to \psi, \phi \to \psi, \Delta'_w, \kappa^*\)}
  %  \RightLabel{\(\ctr\)}
  %  \UnaryInfC{\( \Gamma_w \gg \phi \to \psi, \Delta'_w, \kappa^*\)}
  %  \DisplayProof
  %\]
  %where \(\kappa_w = \kappa\), \({\gg}_w = {\gg}\) and \(\Delta_w = \phi \to \psi, \Delta'_w\).

  Case \(w\) is \((\toLsat)\) or \((\toRsat)\). Analogous to the \((\Ssat)\) case.
  %\[
  %  \AxiomC{\(w0\)}
  %  \noLine
  %  \UnaryInfC{\(\kappa_0 : \phi \to \psi, \Gamma'_w \gg \phi, \Delta_w\)}
  %  \AxiomC{\(w1\)}
  %  \noLine
  %  \UnaryInfC{\(\kappa_1 : \psi,\phi \to \psi, \Gamma'_w \gg  \Delta_w\)}
  %  \RightLabel{\(\toLsat\)}
  %  \BinaryInfC{\(\kappa_0 \vee \kappa_1 : \phi \to \psi, \Gamma'_w \gg \Delta_w\)}
  %  \DisplayProof
  %\]
  %where \(\kappa_w = \kappa_0 \vee \kappa_1\), \({\gg}_w = {\gg}\) and \(\Gamma_w = \phi \to \psi, \Gamma'_w\).
  %\(\alpha(w)\) is defined as
  %\[
  %  \AxiomC{\(\alpha(w0)\)}
  %  \noLine
  %  \UnaryInfC{\(\phi \to \psi, \Gamma'_w \gg \phi, \Delta_w, \kappa^*_0 \)}
  %  \RightLabel{\(\wk + \veeR\)}
  %  \doubleLine
  %  \UnaryInfC{\(\phi \to \psi, \Gamma'_w \gg \phi, \Delta_w, \kappa^*_0 \vee \kappa^*_1\)}
  %  \AxiomC{\(\alpha(w1)\)}
  %  \noLine
  %  \UnaryInfC{\(\psi,\phi \to \psi, \Gamma'_w \gg  \Delta_w, \kappa^*_1\)}
  %  \RightLabel{\(\wk + \veeR\)}
  %  \doubleLine
  %  \UnaryInfC{\(\psi,\phi \to \psi, \Gamma'_w \gg  \Delta_w, \kappa^*_0 \vee \kappa^*_1\)}
  %  \RightLabel{\(\toL\)}
  %  \BinaryInfC{\(\phi \to \psi, \phi \to \psi, \Gamma'_w \gg \Delta_w, \kappa^*_0 \vee \kappa^*_1\)}
  %  \RightLabel{\(\ctr\)}
  %  \UnaryInfC{\(\phi \to \psi, \Gamma'_w \gg \Delta_w, \kappa^*_0 \vee \kappa^*_1\)}
  %  \DisplayProof
  %\]

  Case \(w\) is \((\rep)\).
  The function is defined as
  \[
    \AxiomC{\(\)}
    \RightLabel{\(\rep\)}
    \UnaryInfC{\(x_w : \Gamma_w \gg_w \Delta_w\)}
    \DisplayProof
    \quad
    \overset{\alpha}{\longmapsto}
    \quad
    \AxiomC{\(\alpha(w^\circ)\)}
    \noLine
    \UnaryInfC{\(\Gamma_w \gg_w \Delta_w, \kappa_{w^\circ}^*\)}
    \RightLabel{\(\wk\)}
    \UnaryInfC{\(\Gamma_w \gg_w \Delta_w, x^*_w, \kappa_{w^\circ}^*\)}
    \AxiomC{\(\tau\)}
    \noLine
    \UnaryInfC{\(\kappa_{w^\circ}^*, \Gamma_w \gg_w \Delta_w, x^*_w\)}
    \RightLabel{\(\cut\)}
    \BinaryInfC{\(\Gamma_w \gg_w \Delta_w, x^*_w\)}
    \DisplayProof
  \]
  where \(\tau\) is a proof of \(\kappa_{w^\circ}^*, \Gamma_w \gg_w \Delta_w, x^*_w\) in \(\n{\GLS}\) which exists since \(\GLS \vdash x^*_w \leftrightarrow \kappa^*_{w^\circ}\).

  Case \(w\) is \((\interpolation{\KT})\).
  Then \(w\) is of shape
  \[
    \AxiomC{\(\begin{matrix} w^{\nec} \\ \kappa^{\nec} : \necd \Sigma^s \Rightarrow \end{matrix}\)}
    \AxiomC{\(\left[\begin{matrix} w^{\pos}_\phi \\ \kappa^{\pos}_\phi : \necd \Sigma^s \Rightarrow \phi\end{matrix}\right]_{\phi \in \Theta}\)}
    \RightLabel{\(\interpolation{\KT}\)}
    \BinaryInfC{\(\nec \kappa^{\nec} \wedge \bigwedge_{\phi \in \Theta} \pos \kappa^{\pos}_\phi \wedge \bigwedge (\Gamma'_w \cap V_+) \wedge \neg (\Delta'_w \cap V_-) : \nec \Sigma, \Gamma'_w \gg \nec \Theta, \Delta'_w\)}
    \DisplayProof
  \]
  where \(\kappa_w =\nec \kappa^{\nec} \wedge \bigwedge_{\phi \in \Theta} \pos \kappa^{\pos}_\phi \wedge \bigwedge (\Gamma'_w \cap V_+) \wedge \neg (\Delta'_w \cap V_-)\),  \({\gg_w} = {\gg}\), \(\Gamma_w = \nec \Sigma, \Gamma'_w\) and \(\Delta_w = \nec \Theta, \Delta'_w\).
  Since \((\cdot)^* : B_T \longrightarrow \mathcal{L}_{\nec}\) and \(B_T \inter \text{Var} = \varnothing\) we have that
  \[
    \kappa^*_w = \nec (\kappa^{\nec})^* \wedge \bigwedge_{\phi \in \Theta} \pos (\kappa^{\pos}_\phi)^* \wedge \bigwedge (\Gamma'_w \cap V_+) \wedge \neg (\Delta'_w \cap V_-).
  \]
  So to build a preproof of \(\Gamma_w \gg_w \Delta, \kappa^*_w\) we build a preproof for each of the conjuncts and join them using \((\wedgeR)\).
  \begin{itemize}
    \item Proof for \(\nec\kappa^{\nec}\) and \(\pos \kappa^{\nec}_\phi\) for \(\phi \in \Theta\).
      \(\alpha(w)\) is respectively:
      \[
        \AxiomC{\(\alpha(w^{\nec})\)}
        \noLine
        \UnaryInfC{\(\necd \Sigma^s \Rightarrow (\kappa^{\nec})^*\)}
        \RightLabel{\(\modal{\KT}\)}
        \UnaryInfC{\(\nec \Sigma, \Gamma'_w \gg \nec \Theta, \Delta'_w, \nec(\kappa^{\nec})^*\)}
        \DisplayProof
        \qquad
        \AxiomC{\(\alpha(w^{\pos}_\phi)\)}
        \noLine
        \UnaryInfC{\(\necd \Sigma^s \Rightarrow \phi, (\kappa^{\pos}_\phi)^*\)}
        \RightLabel{\(\negL + \wk\)}
        \UnaryInfC{\(\necd\neg (\kappa^{\pos}_\phi)^*, \necd \Sigma^s \Rightarrow \phi \)}
        \RightLabel{\(\modal{\KT}\)}
        \UnaryInfC{\(\nec \neg (\kappa^{\pos}_\phi)^*, \nec \Sigma, \Gamma'_w \gg \nec \Theta, \Delta'_w\)}
        \RightLabel{\(\negR\)}
        \UnaryInfC{\(\nec \Sigma, \Gamma'_w \gg \nec \Theta, \Delta'_w, \pos(\kappa^{\pos}_\phi)^*\)}
        \DisplayProof
      \]
    \item Proof for \(p \in \Gamma'_w \inter V_+\) and \(\neg q\) where \(q \in \Delta'_w \inter V_-\).
      In each case the preproof is, respectively:
      \[
        \AxiomC{\(\)}
        \RightLabel{\(\ax\)}
        \UnaryInfC{\( \nec \Sigma, \Gamma'_w \gg \nec \Theta, \Delta'_w, p\)}
        \DisplayProof
        \qquad
        \AxiomC{\(\)}
        \RightLabel{\(\ax\)}
  \UnaryInfC{\(q, \nec \Sigma, \Gamma'_w \gg \nec \Theta, \Delta'_w\)}
        \RightLabel{\(\negR\)}
        \UnaryInfC{\(\nec \Sigma, \Gamma'_w \gg \nec \Theta, \Delta'_w, \neg q\)}
        \DisplayProof
      \]
  \end{itemize}

  To each node \(w\) assign the measure \((|\Gamma_w \gg_w \Delta_w|_{\mathrm{sat}}, \lgh(w))\) with the lexicographic order.
  We notice that in each case the measure decreases from \(w\) to its corecursive calls, except when \(w\) is \((\interpolation{\KT})\).
  However, in this case there is an application of \(\modal{\KT}\) from the root of \(\alpha(w)\) to the corecursive calls, guaranteeing progress.

\paragraph{Definition of \(\beta\) at Lemma~\ref{lm:second-verification}.}
  We will define \(\beta\) corecursively.
  That \(\beta(w, \pi)\) is a preproof will be trivial from construction, when the definition is finished we will show that it is indeed a proof.
  We proceed by cases on the shape of \(w\).

  Case \(w\) is \((\ax)\) or \((\botL)\).
  \(\beta(w,\pi)\) is obtained by applying rule \((\botL)\).

  Case \(w\) is \((\emp)\).
  \(\beta(w,\pi)\) is obtained from \(\pi\) by applying the weakening rule.

  Case \(w\) is \((\Ssat)\).
  Then the function is defined as
  \[
    \left(\AxiomC{\(w0\)}
    \noLine
    \UnaryInfC{\(\kappa : \necd \phi, \Gamma'_w \Rrightarrow \Delta_w\)}
    \RightLabel{\(\Ssat\)}
    \UnaryInfC{\(\kappa : \nec \phi, \Gamma'_w \Rrightarrow \Delta_w\)}
    \DisplayProof,
    \quad
    \AxiomC{\(\pi\)}
    \noLine
    \UnaryInfC{\(\nec \phi, \Gamma'_w, \Phi \Rrightarrow \Delta_w, \Psi\)}
    \DisplayProof\right)
    \qquad
    \overset{\beta}{\longmapsto}
    \qquad
    \AxiomC{\(\beta(w0,\wk(\pi))\)}
    \noLine
    \UnaryInfC{\(\kappa^*, \Phi \Rrightarrow \Psi\)}
    \RightLabel{\(\wk\)}
    \UnaryInfC{\(\kappa^*, \Phi \Rrightarrow \Psi\)}
    \DisplayProof
  \]
  where \(\kappa_w = \kappa\), \({\gg}_w = {\Rrightarrow}\) and \(\Gamma_w = \nec\phi, \Gamma'_w\).

  Case \(w\) is \((\toLsat)\) or \((\toRsat)\). Analogous to the \((\Ssat)\) case.

  Case \(w\) is \((\rep)\).
  The function is defined as
  \[
    \AxiomC{\(\)}
    \RightLabel{\(\rep\)}
    \UnaryInfC{\(x_w : \Gamma_w \gg_w \Delta_w\)}
    \DisplayProof, \quad
    \AxiomC{\(\pi\)}
    \noLine
    \UnaryInfC{\(\Gamma_w, \Phi \gg_w \Delta_w, \Psi\)}
    \DisplayProof
    \quad
    \overset{\beta}{\longmapsto}
    \quad
    \AxiomC{\(\tau\)}
    \noLine
    \UnaryInfC{\(x^*_w, \Phi \gg_w \Psi, \kappa_{w^\circ}^*\)}
    \AxiomC{\(\beta(w^\circ, \pi)\)}
    \noLine
    \UnaryInfC{\(\kappa^*_{w^\circ}, \Phi \gg_w \Psi\)}
    \RightLabel{\(\wk\)}
    \UnaryInfC{\(\kappa^*_{w^\circ},x^*_w, \Phi \gg_w \Psi\)}
    \RightLabel{\(\cut\)}
    \BinaryInfC{\( x^*_w, \Phi \gg_w \Psi\)}
    \DisplayProof
  \]
  where \(\tau\) is a proof of \(x^*_w, \Phi \gg_w \Psi, \kappa_{w^\circ}^*\) in \(\n{\GLS}\) which exists since \(\GLS \vdash x^*_w \leftrightarrow \kappa^*_{w^\circ}\).

  %Case \(w\) is \((\toRsat)\).
  %Then the function is defined as
  %\[
  %  \left(
  %  \AxiomC{\(w0\)}
  %  \noLine
  %  \UnaryInfC{\(\kappa : \phi, \Gamma_w \gg \psi, \phi \to \psi, \Delta'_w\)}
  %  \RightLabel{\(\toRsat\)}
  %  \UnaryInfC{\(\kappa : \Gamma_w \gg \phi \to \psi, \Delta'_w\)}
  %  \DisplayProof,
  %  \quad
  %  \AxiomC{\(\pi\)}
  %  \noLine
  %  \UnaryInfC{\(\kappa : \Gamma_w, \Phi \gg \phi \to \psi, \Delta'_w, \Psi\)}
  %  \DisplayProof
  %  \right)
  %  \quad
  %  \overset{\beta}{\longmapsto}
  %  \quad
  %  \AxiomC{\(\beta(w0,\wk(\pi))\)}
  %  \noLine
  %  \UnaryInfC{\(\kappa^*, \Phi \gg \Psi\)}
  %  \RightLabel{\(\wk\)}
  %  \UnaryInfC{\(\kappa^*, \Phi \gg \Psi\)}
  %  \DisplayProof
  %\]
  %where \(\kappa_w = \kappa\), \({\gg}_w = {\gg}\) and \(\Delta_w = \phi \to \psi, \Delta'_w\).

  Case \(w\) is \((\interpolation{\KT})\).
  Then \(w\) is of shape
  \[
    \AxiomC{\(\begin{matrix} w^{\nec} \\ \kappa^{\nec} : \necd \Sigma^s \Rightarrow \end{matrix}\)}
    \AxiomC{\(\left[\begin{matrix} w^{\pos}_\phi \\ \kappa^{\pos}_\phi : \necd \Sigma^s \Rightarrow \phi\end{matrix}\right]_{\phi \in \Theta}\)}
    \RightLabel{\(\interpolation{\KT}\)}
    \BinaryInfC{\(\nec \kappa^{\nec} \wedge \bigwedge_{\phi \in \Theta} \pos \kappa^{\pos}_\phi \wedge \bigwedge (\Gamma'_w \cap V_+) \wedge \neg (\Delta'_w \cap V_-) : \nec \Sigma, \Gamma'_w \gg \nec \Theta, \Delta'_w\)}
    \DisplayProof
  \]
  where \(\kappa_w =\nec \kappa^{\nec} \wedge \bigwedge_{\phi \in \Theta} \pos \kappa^{\pos}_\phi \wedge \bigwedge (\Gamma'_w \cap V_+) \wedge \neg (\Delta'_w \cap V_-)\),  \({\gg_w} = {\gg}\), \(\Gamma_w = \nec \Sigma, \Gamma'_w\) and \(\Delta_w = \nec \Theta, \Delta'_w\).
  %We note that since \((\cdot)^* : B_T \longrightarrow \mathcal{L}_{\nec}\) and \(B_T \inter \text{Var} = \varnothing\) we have that
  %\[
  %  \kappa^*_w = \nec (\kappa^{\nec})^* \wedge \bigwedge_{\phi \in \Theta} \pos (\kappa^{\pos}_\phi)^* \wedge \bigwedge (\Gamma'_w \cap V_+) \wedge \neg (\Delta'_w \cap V_-).
  %\]
  We proceed by further case analysis in the last rule applied to \(\pi\).
  \begin{itemize}
    \item Last rule is \(\ax\). We have that a propositional variable \(p\) must occur at the left and right side of the sequent. Since \((\Gamma'_w \inter \Delta'_w) \inter \text{Var} = \varnothing\) we have three options.
      If \(p \in \Phi \inter \Psi\) then the desired proof is just \(\ax\).
      If \(p \in \Gamma'_w \inter \Psi\), then \(p \in \Psi\) implies that \(p \in V_+\) and the desired proof is using that \(p\) is a conjunct of \(\kappa^*_w\).
      If \(p \in \Phi \inter \Delta'_w\), then \(p \in \Phi\) implies that \(p \in V_-\) and the desired proof is using that \(\neg p\) is a conjunct of \(\kappa^*_w\).
    \item Last rule is \(\botL\).
      Since \(\bot \not \in \Gamma'_w\) (by definition of \(\interpolation{\KT}\)), it must be the case that \(\bot \in \Phi\).
      Then the desired preproof is obtained using \(\botL\).
    \item Last rule is \(\botR\).
      If \(\bot \in \Delta'_w\) then \(\Delta'_w = \bot, \Delta''_w\). 
      We define \(\beta\) as
      \[
        \left(w,\quad
          \AxiomC{\(\pi_0\)}
          \noLine
          \UnaryInfC{\(\nec \Sigma, \Gamma'_w, \Phi \gg \nec \Theta, \Delta''_w, \Psi\)}
          \RightLabel{\(\botR\)}
          \UnaryInfC{\(\nec \Sigma, \Gamma'_w, \Phi \gg \nec \Theta, \bot, \Delta''_w, \Psi\)}
          \DisplayProof
          \right)
          \quad
          \overset{\beta}{\longmapsto}
          \quad
          \AxiomC{\(\beta(w,\wk(\pi_0))\)}
          \noLine
          \UnaryInfC{\(\kappa^*_w, \Phi \gg \Psi\)}
          \RightLabel{\(\wk\)}
          \UnaryInfC{\(\kappa^*_w, \Phi \gg \Psi\)}
          \DisplayProof
      \]
      If \(\bot \not \in \Delta'_w\) then \(\bot \in \Psi\) so \(\Psi = \bot, \Psi'\).
      We define \(\beta\) as
      \[
        \left(w,\quad
          \AxiomC{\(\pi_0\)}
          \noLine
          \UnaryInfC{\(\nec \Sigma, \Gamma'_w, \Phi \gg \nec \Theta, \Delta'_w, \Psi'\)}
          \RightLabel{\(\botR\)}
          \UnaryInfC{\(\nec \Sigma, \Gamma'_w, \Phi \gg \nec \Theta, \Delta''_w, \bot, \Psi'\)}
          \DisplayProof
          \right)
          \quad
          \overset{\beta}{\longmapsto}
          \quad
          \AxiomC{\(\beta(w,\pi_0)\)}
          \noLine
          \UnaryInfC{\(\Phi \gg \Psi'\)}
          \RightLabel{\(\botR\)}
          \UnaryInfC{\(\Phi \gg \bot,\Psi'\)}
          \DisplayProof
      \]

    \item Last rule is \((\S)\). If the principal formula \(\nec\phi\) of \((\S)\) is in \(\nec \Sigma\), then we know by saturation that \(\phi \in \nec \Sigma, \Gamma'_w\). Let \(\Sigma' = \phi, \Sigma\), we define \(\beta\) as
      \[
        \left(
          w,
          \AxiomC{\(\pi_0\)}
          \noLine
          \UnaryInfC{\(\necd \phi, \nec \Sigma', \Gamma'_w, \Phi \Rrightarrow \nec \Theta, \Delta'_w, \Psi\)}
          \RightLabel{\(\S\)}
          \UnaryInfC{\(\nec \phi, \nec \Sigma', \Gamma'_w, \Phi \Rrightarrow \nec \Theta, \Delta'_w, \Psi\)}
          \DisplayProof
        \right)
        \overset{\beta}{\longmapsto}
        \AxiomC{\(\beta(w,\ctr(\pi_0))\)}
        \noLine
        \UnaryInfC{\(\kappa^*_w, \Phi \Rrightarrow \Psi\)}
        \RightLabel{\(\wk\)}
        \UnaryInfC{\(\kappa^*_w, \Phi \Rrightarrow \Psi\)}
        \DisplayProof
      \]
      where we contract \(\pi_0\) to get rid of one instance of \(\phi\) at the left.
      If the principal formula \(\nec \phi\) of \((\S)\) is not in \(\nec \Sigma\), it must be in \(\Phi\), so let \(\Phi = \nec \phi, \Phi'\).
      We define \(\beta\) as
      \[
        \left(
          w,
          \AxiomC{\(\pi_0\)}
          \noLine
          \UnaryInfC{\(\nec \Sigma, \Gamma'_w, \necd \phi, \Phi' \Rrightarrow \nec \Theta, \Delta'_w, \Psi\)}
          \RightLabel{\(\S\)}
          \UnaryInfC{\(\nec \Sigma, \Gamma'_w, \nec\phi,\Phi' \Rrightarrow \nec \Theta, \Delta'_w, \Psi\)}
          \DisplayProof
        \right)
        \overset{\beta}{\longmapsto}
        \AxiomC{\(\beta(w,\pi_0)\)}
        \noLine
        \UnaryInfC{\(\kappa^*_w, \necd \phi, \Phi' \Rrightarrow \Psi\)}
        \RightLabel{\(\S\)}
        \UnaryInfC{\(\kappa^*_w, \nec \phi, \Phi' \Rrightarrow \Psi\)}
        \DisplayProof
      \]

    \item Last rule is \((\toL)\) or \((\toR)\). Analogous to \((\S)\) case (with cases on where the principal formula occurs and possibly needed to use admissibility of contraction and weakening at proof \(\pi\)).

    \item Last rule is \(\modal{\KT}\).
      If the principal formula \(\nec \phi\) of \(\modal{\KT}\) is in \(\nec \Theta\), then let \(\nec \Theta = \nec \phi, \nec \Theta'\).
      We define \(\beta\) as
      \[
        \left(
          w,
          \quad
          \AxiomC{\(\pi_0\)}
          \noLine
          \UnaryInfC{\(\necd \Sigma', \necd \Phi_{\nec} \Rightarrow \phi\)}
          \RightLabel{\(\modal{\KT}\)}
          \UnaryInfC{\(\nec \Sigma, \Gamma'_w, \Phi \gg \nec \phi, \nec \Theta', \Delta'_w, \Psi\)}
          \DisplayProof
        \right)
        \quad
        \overset{\beta}{\longmapsto}
        \quad
        \AxiomC{\(\beta\left(w^{\pos}_\phi, \wk(\pi_0)\right)\)}
        \noLine
        \UnaryInfC{\((\kappa^{\pos}_\phi)^*, \necd \Phi_{\nec} \Rightarrow \)}
        \RightLabel{\(\negR\)}
        \UnaryInfC{\(\necd \Phi_{\nec} \Rightarrow \neg (\kappa^{\pos}_\phi)^*\)}
        \RightLabel{\(\modal{\KT}\)}
        \UnaryInfC{\(\Phi \gg \nec \neg (\kappa^{\pos}_\phi)^*, \Psi\)}
        \RightLabel{\(\negL\)}
        \UnaryInfC{\(\pos (\kappa^{\pos}_\phi)^*, \Phi \gg  \Psi\)}
        \doubleLine
        \RightLabel{\(\wk + \wedgeL\)}
        \UnaryInfC{\(\kappa^*_w, \Phi \gg  \Psi\)}
        \DisplayProof
      \]
      where \(\Sigma' \subseteq \Sigma\), \(\nec\Phi_{\nec} \subseteq \Phi\) and we applied weakening on \(\pi_0\) to obtain the full \(\necd \Sigma\) instead of \(\necd \Sigma'\).
      If the principal formula \(\nec \phi\) of \(\modal{\KT}\) is in \(\Psi\), then let \(\Psi = \nec \phi, \Psi'\).
      We define \(\beta\) as
      \[
        \left(
          w,
          \quad
          \AxiomC{\(\pi_0\)}
          \noLine
          \UnaryInfC{\(\necd \Sigma', \necd \Phi_{\nec} \Rightarrow \phi\)}
          \RightLabel{\(\modal{\KT}\)}
          \UnaryInfC{\(\nec \Sigma, \Gamma'_w, \Phi \gg  \nec \Theta, \Delta'_w, \nec \phi,\Psi'\)}
          \DisplayProof
        \right)
        \quad
        \overset{\beta}{\longmapsto}
        \quad
        \AxiomC{\(\beta(w^{\nec}, \wk(\pi_0))\)}
        \noLine
        \UnaryInfC{\((\kappa^{\nec})^*, \necd \Phi_{\nec} \Rightarrow \phi\)}
        \RightLabel{\(\wk\)}
        \UnaryInfC{\(\necd (\kappa^{\nec})^*, \necd \Phi_{\nec} \Rightarrow \phi\)}
        \RightLabel{\(\modal{\KT}\)}
        \UnaryInfC{\(\nec(\kappa^{\nec})^*, \Phi \gg \nec\phi, \Psi'\)}
        \doubleLine
        \RightLabel{\(\wk + \wedgeL\)}
        \UnaryInfC{\(\kappa^*_w, \Phi \gg  \nec \phi, \Psi'\)}
        \DisplayProof
      \]
      where \(\Sigma' \subseteq \Sigma\), \(\nec\Phi_{\nec} \subseteq \Phi\) and we applied weakening on \(\pi_0\) to obtain the full \(\necd \Sigma\) instead of \(\necd \Sigma'\).
  \end{itemize}
  To each pair \((w,\pi)\) assign the measure \((|\Gamma_w \gg_w \Delta_w|_{\mathrm{sat}}, \lgh(w), \lhg(\pi))\) with the lexicographic order.
  We notice that in each case the measure decreases from \((w,\pi)\) to its corecursive calls, except when \(w\) is \((\interpolation{\KT})\) and \(\pi\) ends in the \((\modal{\KT})\) rule.
  However, in this case there is an application of \((\modal{\KT})\) from the root of \(\alpha(w)\) to the corecursive calls, guaranteeing progress.

\nocite{*}
\bibliographystyle{eptcs}
\bibliography{generic}
\end{document}